\documentclass[10pt]{article}
 
\usepackage{amsmath}
\usepackage{amsfonts}
\usepackage{array}
\usepackage{enumerate}
\usepackage{graphicx}
\usepackage{amssymb}
\usepackage{amsthm}
\usepackage[version=4]{mhchem}
\usepackage{xcolor}
\usepackage{chemarr}
\usepackage{mathtools}
\usepackage{empheq}
\usepackage[margin=1in]{geometry}
\usepackage{soul}
\usepackage[numbers,square,sort&compress]{natbib}
\newtheorem{proposition}{Proposition}

\newtheorem{definition}{Definition}
\newtheorem{remark}{Remark}
\newtheorem{example}{Example}
\newcommand{\norm}[2][\relax]{\ifx#1\relax \ensuremath{\left\Vert#2\right\Vert} \else \ensuremath{\left\Vert#2\right\Vert_{#1}}\fi}
\newcommand\smallO{
  \mathchoice
    {{\scriptstyle\mathcal{O}}}% \displaystyle
    {{\scriptstyle\mathcal{O}}}% \textstyle
    {{\scriptscriptstyle\mathcal{O}}}% \scriptstyle
    {\scalebox{.7}{$\scriptscriptstyle\mathcal{O}$}}%\scriptscriptstyle
  }
\begin{document}

\title{Timescale disparity and the reduction of the chemical master equation: A geometric approach via the linear noise approximation}

\author{Justin Eilertsen\\
            Mathematical Reviews\\ American Mathematical Society\\
            416 4th Street\\Ann Arbor, MI, 48103\\
            email: {\tt jse@ams.org}\\\\
        Wylie Stroberg\\
          Department of Mechanical Engineering \\ 
          Department of Biomedical Engineering \\
          University of Alberta, Edmonton, Alberta, Canada\\
          email: stroberg@ualberta.ca}

%\date{}

\maketitle
\begin{abstract}
The chemical master equation dictates the stochastic behavior of chemical reaction networks in small volumes that are well-mixed and homogeneous. Although the Gillespie algorithm is capable of generating exact realizations of the master equation, it is often computationally expensive, especially when the network exhibits disparate reaction rates and timescales. Consequently, model reduction techniques are frequently employed to reduce the computational demands of the stochastic simulation algorithm while preserving accurate timecourse statistics. Often, the reduction of the master equation is facilitated through the adaptation of deterministic techniques from singular perturbation theory, resulting in homologous -- yet fundamentally heuristic -- reduced master equations. Such heuristic reductions are often accurate when applied to linear reaction networks but inaccurate when applied to nonlinear reaction networks. The failure of heuristic reductions is not fully understood, and it remains unclear when and why the heuristic reduction of the master equation will succeed. In this work, we take a significant step towards bridging this divide by proving that every first-order reaction network admits an accurate heuristic reduction of its associated master equation provided a specific geometric criterion holds. We also discuss the implications of this result as it pertains to nonlinear reaction networks. Specifically, we explain, through the lens of geometric perturbation theory, why the heuristically reduced CME of the nonlinear Michaelis-Menten reaction network loses precision when substrate concentrations are moderate.
\end{abstract}

%%%%%%%%%%
\section{Introduction: Motivation and Problem Statement}

Chemical reaction networks are typically modeled with an ordinary differential equation (ODE) system based on the law of mass action. However, mass action ODE models are generally nonlinear, rendering their analysis much more difficult and thereby obfuscating the origin of certain kinetic relationships.

One way to lessen the complexity of a reaction network is to look for (or derive) an approximation that limits the number of dependent variables without compromising the qualitative features of the network's timecourse dynamics. Such reduced models are attainable when the rates of the elementary reactions that comprise the network are disparate. In such scenarios, the network will evolve on multiple fast and slow timescales. Mathematically, this implies that the eigenvalues of an appropriate Jacobian matrix will be widely separated and exhibit a {\it spectral gap}.

The existence of a spectral gap implies the existence of an invariant submanifold that influences the slow dynamics. Center manifold theory~\cite{RobertsBook}, as well as singular perturbation theory~\cite{Lapuz,Jones1995}, capitalize on the existence of invariant manifolds and provide a rigorous framework through which the flow on the invariant manifold can be approximated. The advantage offered by center manifold and singular perturbation techniques is that longtime dynamics are approximated by projecting the system onto the lower-dimensional and invariant submanifold, resulting in a reduction in the number of variables (coordinates). The resulting low-dimensional system is often more amenable to analysis than the mass action model and is sometimes even linear. 

As a didactic example, consider the reaction discussed in~\citet{Janssen1989}
\begin{align}
\ce{$\bar{X}$ & <=>[$\varepsilon k_1$][$k_{2}$] $\bar{Y}$ ->[$k_3$] $\bar{P}$}. \label{reac1}
\end{align}

The mass action system that describes the temporal evolution of $x,y$ and $p$ (the respective concentrations of $\bar{X},\bar{Y}$ and $P$) is 
\begin{subequations}\label{sys1}
\begin{align}
\dot{x} &= -\varepsilon k_1 x + k_2y\label{xdot}\\
\dot{y} &= \;\;\;\varepsilon k_1x - k_2y - k_3y,\label{ydot}\\
\dot{p}& = \;\;\;k_3y,\label{zdot}
\end{align}
\end{subequations}
where ``$\dot{\phantom{x}}$" denotes differentiation with respect to time, and $k_1,k_2$ and $k_3$ are rate constants. The parameter $\varepsilon$ is intended to reflect the notion that $k_1$ is {\it small}, i.e., $k_1\ll \min\{k_2,k_3\}$. Furthermore, we need not be concerned with \eqref{zdot}, as it decouples from \eqref{xdot} and \eqref{ydot}. Although \eqref{sys1} is linear and admits closed-form solutions, the Jacobian evaluated at the stationary point that coincides with the origin $(x,y)=(0,0)$ reveals a spectral gap, and therefore phase-plane trajectories will ultimately approach the origin in the direction of the slow eigenvector. In this example, the slow invariant submanifold {\it is} the slow eigenvector, and the flow on the slow eigenvector can be approximated via singular perturbation theory. Omitting the details at this point, the singular perturbation reduction of \eqref{sys1} is
\begin{subequations}
\begin{align}
\dot{x} &= -\cfrac{\varepsilon k_1k_3}{k_2+k_3}\cdot x\label{rxdot}=-\dot{p},\\
\dot{y} &=0.
\end{align}
\end{subequations}
Simply stated, the differential equation \eqref{rxdot} approximates the dynamics on the slow eigenvector to leading order in $\varepsilon$. While reducing a linear system to a lower-dimensional linear system via projection onto an invariant subspace may seem redundant, the real power of center manifold (and by extension singular perturbation) theory becomes apparent when the reaction network is nonlinear.

In deterministic systems, timescale disparity may be inherent due to preexisting thermodynamic constraints, or it can be manufactured by limiting the concentration of certain chemical species within an experimental assay~\cite{Segel1988,Segel1989}. But what about in vivo reactions that occur within cells? In extremely small volumes, or when molecular copy numbers are finite, the effect of intrinsic noise cannot be ignored, and therefore we must rely on stochastic models. In an ideal scenario, in which the chemical species of a network are well-mixed so that spatial heterogeneity is negligible, the chemical master equation is the appropriate model to employ. Thus, as we reduce the size of the chemical system, we move from continuous time/continuous space (mass action equations) to continuous time/discrete space (chemical master equations). The solution to the master equation gives the probability distribution, $Pr(\cdot,t)$, over the discrete state space as a function of time.  

The master equation associated with \eqref{sys1} is 
\begin{equation}\label{CME1}
\cfrac{\partial}{\partial t}Pr(n,m,q,t) = \left[\varepsilon k_1(E_n^{+1}E_m^{-1}-1)n + k_2(E_n^{-1}E_m^{+1}-1)m+ k_3(E_m^{+1}-1)m\right]Pr(n,m,q,t)
\end{equation}
where $n$ denotes the number of $\bar{X}$ molecules, $m$ the number of $\bar{Y}$ molecules, $q$ the number of $P$ molecules, and $E_{(\cdot)}^{\pm j}$ are step operators,
\begin{equation*}
E_n^{\pm j}f(n,m) = f(n\pm j,m), \quad E_m^{\pm j}f(n,m) = f(n,m\pm j).
\end{equation*}
If we were to solve \eqref{CME1}, we would obtain the probability of finding the system with ``$n$" $\bar{X}$ molecules, ``$m$" $\bar{Y}$ molecules and ``$q$'' $P$ molecules at time $t$. However, despite the fact that the network \eqref{sys1} is linear, obtaining the solution to the master equation \eqref{CME1} is intractable. Moreover, while the stochastic simulation or ``Gillespie" algorithm~\cite{GILLESPIECME} is capable of generating exact stochastic timecourse realization of CMEs, the algorithm can be computationally expensive when disparate timescales are present~\cite{TURNER200,Sanft}. This raises the question of whether the CME is reducible in a way that is homologous to deterministic reduction \eqref{rxdot}.

In this specific example, the answer is affirmative: the reduced master equation
\begin{equation}\label{rCME1}
\cfrac{\partial}{\partial t}Pr(n,t) = \cfrac{\varepsilon k_1k_3}{k_2+k_3}(E_n^{+1}-1)nPr(n,t)
\end{equation}
accurately approximates the stochastic timecourse dynamics of $X$ obtained from the complete master equation \eqref{CME1}, and its accuracy was established in the earlier work of~\citet{Janssen1989}; see {\sc figure} \ref{FIG1} for numerical confirmation.
\begin{figure}[bth!]
    \centering
    \includegraphics[scale=0.50]{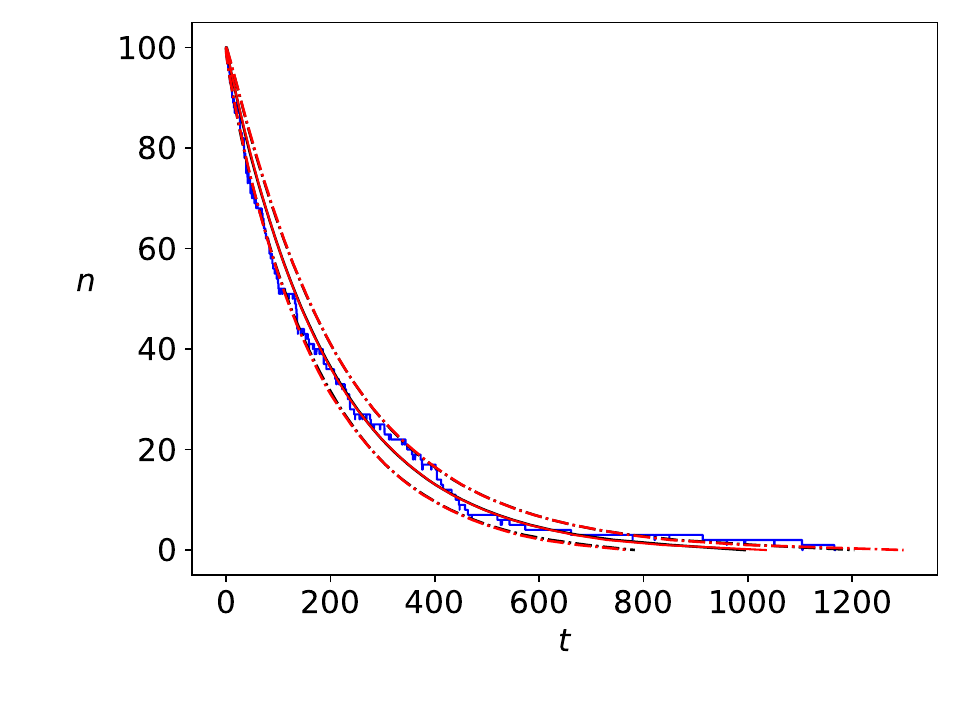}
    \includegraphics[scale=0.50]{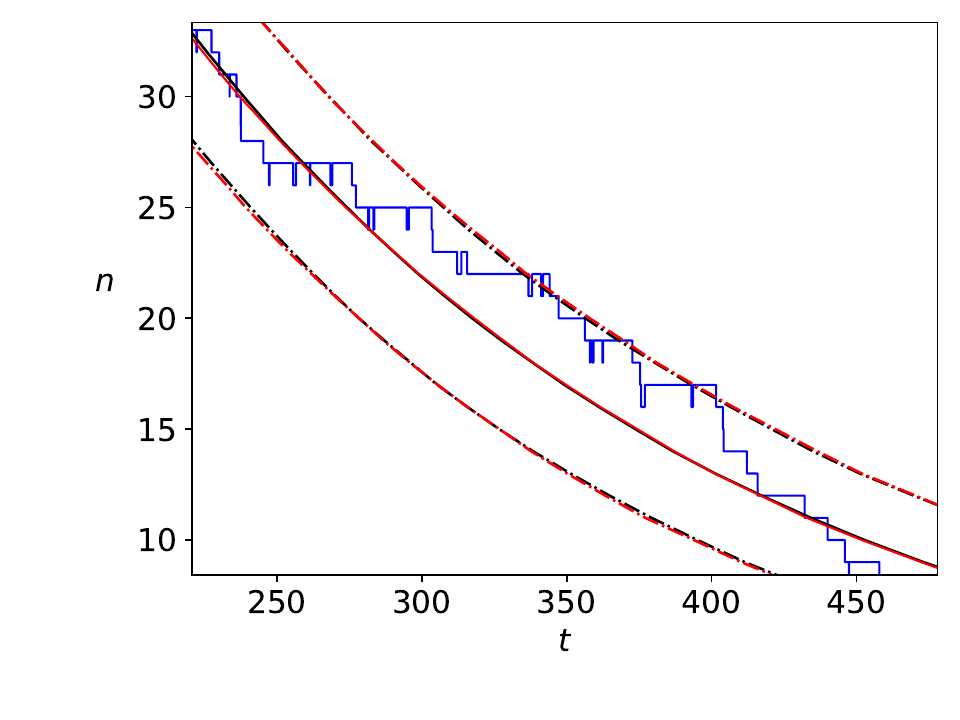}
    \caption{\textbf{The heuristically reduced CME \eqref{rCME1} provides an excellent approximation to the first two moments of $n$, the number of $\bar{X}$ molecules. } In all simulations $k_1=0.01,k_2=1.0$, and $k_3=1.0$ with $n_0=100$, $m_0=q_0=0$, where $n_0,m_0,q_0$ are the initial numbers of $\bar{X}$, $\bar{Y}$ and $P$ molecules. {\sc left}: The solid red curve is the mean computed from $1000$ simulations of the Gillespie algorithm applied to the reduced CME \eqref{rCME1}, and the dashed/dotted red curves are the mean $\pm$ one standard deviation. The solid black line is the mean of $\bar{X}$ computed from $1000$ simulations of the full CME \eqref{CME1} via the Gillespie algorithm, and the dashed/dotted lines are the mean $\pm$ one standard deviation. A single stochastic timecourse trajectory for $\bar{X}$ computed from the full CME via the Gillespie algorithm is presented in blue.  {\sc right}: A close-up of the left panel. Note that the approximated mean and standard deviation for $\bar{X}$ obtained from the reduced CME \eqref{rCME1} is virtually indistinguishable from those computed from the full CME \eqref{CME1}.}
    \label{FIG1}
\end{figure}

The interpretation of the reduced CME \eqref{rCME1} is as follows: given that there are ``$n$" $\bar{X}$ molecules at time $t$, the probability that a single $P$ molecule forms within the infinitesimal window, $[t,t+{\rm d}t)$ is
\begin{equation*}
Pr(n-1,t+{\rm d}t|n,t) =\cfrac{\varepsilon k_1k_3}{k_2+k_3} \;{\rm d}t,
\end{equation*}
and, consequently, we end up with a reduced reaction network,
\begin{align}
\ce{$\bar{X}$ ->[$\kappa$] $P$}, \quad \kappa := \cfrac{\varepsilon k_1k_3}{k_2+k_3}.\label{reac2}
\end{align}
Based on \eqref{reac2}, it is easy to speculate that
\begin{equation}
\cfrac{\partial}{\partial t}Pr(q,t) = \cfrac{\varepsilon k_1k_3}{k_2+k_3}(E_q^{-1}-1)qPr(q,t)
\end{equation}
also accurately estimates the first and second moments of $q$, the number of $P$ molecules. This speculation turns out to be true and is easy to verify numerically; see {\sc figure} \ref{FIG2}.

Observe that the reduced CME \eqref{rCME1} is the exact stochastic analog of the deterministic reduction \eqref{rxdot}. This type of reduction is generally referred to as a heuristic reduction because it can be very difficult, if not impossible, to rigorously derive, especially when the reaction network is nonlinear. However, the attractive advantage it provides is that {\it if} one has an a priori notion of when the heuristic reduction is accurate, then the reduced form of the CME is obtained directly from the analysis of the deterministic mass action equations, thereby circumnavigating the need to analyze the more complicated stochastic model. 

What is unclear from the literature is {\it when} and {\it why} heuristic reductions are reliable, as they can fail; see~\cite{GrimaBreakdown,Agarwal2012,Holehouse} for notable examples. Currently, there is no theorem that clearly outlines when we can expect the heuristic reduction of the CME to accurately and reliably approximate the first and second moments of the full CME. One can speculate, for example, that the accuracy of heuristic reductions is attributable to the linearity of the underlying reaction network. However, this does not explain the reported success of heuristic reductions applied to nonlinear reaction networks~\cite{Rao2003,Mastny2007,Tyson2008,Sanft}. Moreover, it is easy to construct a counterexample that (partially) debunks this hypothesis. Consider the reaction network \eqref{reac1} but this time with small $k_3$ instead of small $k_1$ (i.e. $k_3\mapsto \varepsilon k_3$ instead of $k_1 \mapsto \varepsilon k_1$). The reduced mass action equations are 
\begin{subequations}
\begin{align}
\dot{x} &= -\cfrac{\varepsilon k_3k_1}{k_1+k_2}\cdot x,\label{rxdot2}\\
\dot{y} &= -\cfrac{\varepsilon k_3k_1}{k_1+k_2}\cdot y. \label{rydot2}
\end{align}
\end{subequations}
but in this case the heuristic reduction of the CME,
\begin{equation}\label{rCME2}
\cfrac{\partial}{\partial t}Pr(n,t) = \cfrac{\varepsilon k_3k_1}{k_1+k_2}(E_n^{+1}-1)nPr(n,t)
\end{equation}
does {\it not} provide accurate stochastic timecourse approximations for $n$ (this is easily verified numerically, and we address this in the Appendix) obtained from the corresponding full CME
\begin{equation}\label{CME2}
\cfrac{\partial}{\partial t}Pr(n,m,q,t) = \left[k_1(E_n^{+1}E_m^{-1}-1)n + k_2(E_n^{-1}E_m^{+1}-1)m+ \varepsilon k_3(E_m^{+1}-1)m\right]Pr(n,m,q,t).
\end{equation}
On the other hand, it is straightforward to verify (numerically) that 
\begin{equation}\label{rCME3}
\cfrac{\partial}{\partial t}Pr(q,t) = \cfrac{\varepsilon k_1k_3}{k_1+k_2}(E_q^{-1}-1)qPr(q,t)
\end{equation}
provides an excellent approximation to the first two moments of $Pr(q,t)$; see {\sc figure} \ref{FIG2}.
\begin{figure}[htb!]
    \centering
    \includegraphics[scale=0.50]{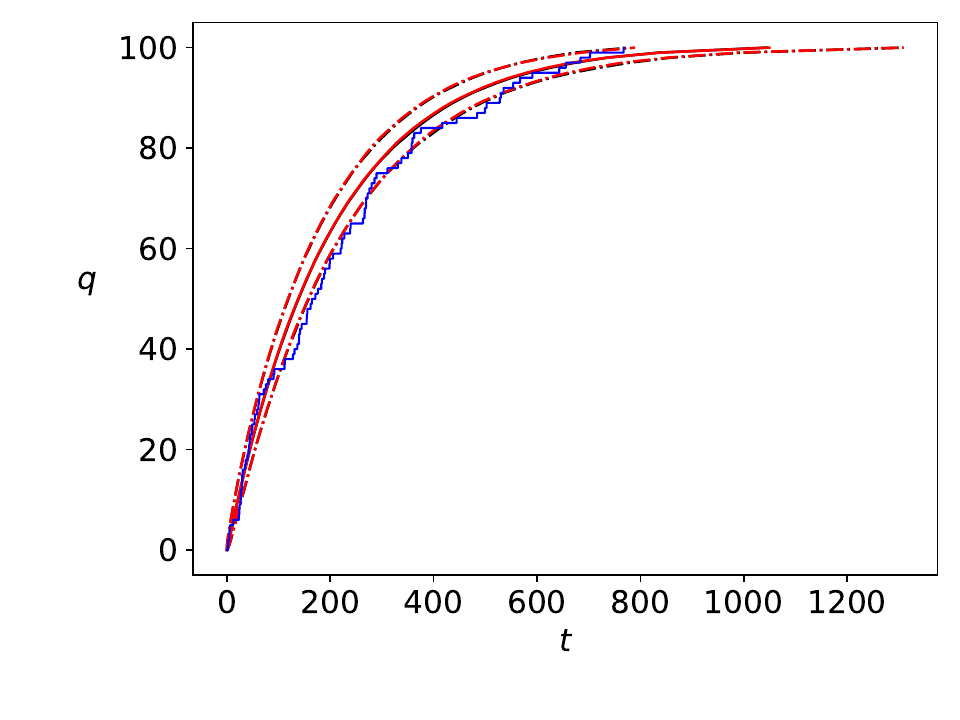}
    \includegraphics[scale=0.50]{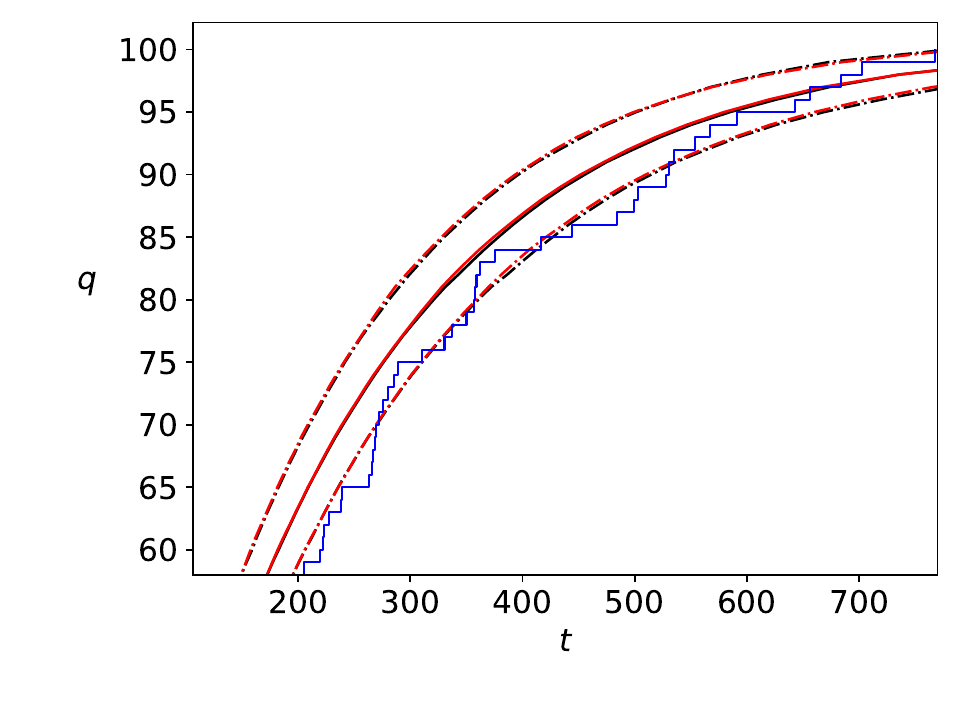}
    \caption{\textbf{The heuristically reduced CME \eqref{rCME3} provides an excellent approximation to the first two moments of $q$, the number of $P$ molecules. } In all simulations $k_1=1.0,k_2=1.0$, and $k_3=0.01$ with $n_0=100$, $m_0=0$ and $q_0=0$, where $n_0,m_0,q_0$ are the initial numbers of $\bar{X}$, $\bar{Y}$ and $P$ molecules. {\sc left}: The solid red curve is the mean computed from $1000$ simulations of the Gillespie algorithm applied to the reduced CME \eqref{rCME2}, and the dashed/dotted red curves are the mean $\pm$ one standard deviation. The solid black line is the mean of $q$ (the copy number of $P$ molecules) computed from $1000$ simulations of the full CME \eqref{CME2} via the Gillespie algorithm, and the dashed/dotted lines are the mean $\pm$ one standard deviation. A single stochastic timecourse trajectory for $P$ computed from the full CME via the Gillespie algorithm is presented in blue.  {\sc right}: A close-up of the left panel. Note that the approximated mean and standard deviation for $P$ obtained from the reduced CME \eqref{rCME3} is virtually indistinguishable from those computed from the full CME \eqref{CME2}.}
    \label{FIG2}
\end{figure}

Thus, even within the context of linear reaction networks, it is unclear when and why certain heuristic reductions hold. In this work, we take a significant step towards understanding the reliability of heuristic reductions, and address the principal question: Under what conditions does the heuristic reduction of the CME accurately approximate both the mean and the variance? Since the CME (even for first order reaction networks) is challenging to analyze directly, we rely on the linear noise approximation (LNA) to understand why heuristic reductions succeed or fail. The motivation behind the utilization of the LNA follows from the well-known fact that the LNA and the CME agree up to second-order moments for zeroth- and first-order reaction networks, as well as for some second-order reaction networks~\cite{GrimaSecondOrder}.

The outline is as follows: In Section \ref{section:2}, we briefly review geometric singular perturbation theory (GSPT), which is the mathematical foundation that justifies the reduced ODE models that are adapted to formulate the heuristically reduced CME. In Section \ref{section:3} we introduce the LNA, its relationship to the CME in terms of zeroth- and first-order reaction networks, and its reduction based on GSPT. In Section \ref{section:4} we prove that the heuristic reduction of the CME always holds for specific zeroth- and first-order reaction networks, as long as the diffusion generated by intrinsic noise is limited to the slow timescale in the linear noise regime. This is the core contribution of our work. In Section \ref{section:5}, we discuss the implications of this result in the context of nonlinear reaction networks, including the well-studied Michaelis-Menten reaction network~\cite{Kang2019,KangKim2017}.

\section{Geometric singular perturbation theory: Linear differential equations}\label{section:2}

Although geometric singular perturbation theory applies to both linear and nonlinear differential equations~\cite{Fenichel1979,HekGSPT}, our interest will be in the former. For our purposes, it will suffice to consider planar linear differential equations, although the results presented in this work extend to higher-dimensional linear equations. 

To motivate the use of singular perturbation methods, consider a linear differential equation of the form
\begin{equation}\label{linear}
\dot{x} = A_0x + \varepsilon A_1 x,
\end{equation}
where $x\in \mathbb{R}^2$ and the matrix $A:=A_0+\varepsilon A_1$ has two distinct, strictly real and negative eigenvalues,
$\lambda_{\pm}(\varepsilon)$, which are analytic in $\varepsilon$. More specifically, we assume that the eigenvalues admit convergent expansions in $\varepsilon$ of the form
\begin{equation}
\lambda_+=\varepsilon \lambda_+^{(1)} + \mathcal{O}(\varepsilon^2),\qquad \lambda_{-}=\lambda_-^{(0)} + \varepsilon \lambda_{-}^{(1)}+\mathcal{O}(\varepsilon^2).
\end{equation}

Although the complete system \eqref{linear} has a well-known solution
\begin{equation*}\
x(t) = e^{\displaystyle t A}x(0),
\end{equation*}
the assumption that $\lambda_-(\varepsilon)\ll \lambda_+(\varepsilon)<0$ ensures that the dynamics of the system \eqref{linear} unfold over two timescales, and this property can be exploited to produce accurate approximations. The slow timescale, $\tau=\varepsilon t$, is the timescale that accounts for changes in the direction parallel to the slow eigenspace: $\lambda_+(\varepsilon)t = \lambda_+^{(1)}\tau+\mathcal{O}(\varepsilon \tau)$. In contrast, the fast timescale, $t$, is the timescale on which changes in directions parallel to the fast eigenspace occur. Thus, \eqref{linear} is an example of a {\it two timescale} problem.

Due to the presence of disparate timescales, we restrict our interest to the long-time behavior of \eqref{linear}, which is determined by the dynamics on the invariant slow eigenspace associated with the eigenvalue $\lambda_+(\varepsilon)$ since all trajectories eventually approach the origin in the direction of the slow eigenspace. We can approximate the long-time dynamics as follows. Let $w(\varepsilon)$ denote the left eigenvector, $w^T(\varepsilon)A =\lambda_+(\varepsilon)w^T(\varepsilon)$. The unique projection matrix $P_{\varepsilon}:\mathbb{R}^2\to {\rm span}\{v_+(\varepsilon)\}$ is
\begin{equation}
P_{\varepsilon} = \cfrac{v_+(\varepsilon)w^T(\varepsilon)}{w^T(\varepsilon)v_+(\varepsilon)}, \quad Av_{\pm}(\varepsilon) = \lambda_{\pm}(\varepsilon)v_{\pm}(\varepsilon).
\end{equation}
The matrices $P_{\varepsilon}$ and $A$ share common eigenvectors and therefore commute
\begin{equation*}
\left[A,P_{\varepsilon}\right] = AP_{\varepsilon}-P_{\varepsilon}A=0.
\end{equation*}
More specifically, $A = \lambda_+(\varepsilon)P_{\varepsilon}+\lambda_-(\varepsilon)(I-P_{\varepsilon})$, and the long-time dynamics of \eqref{linear} is 
\begin{equation}\label{linred}
x(t) = e^{\displaystyle \lambda_+(\varepsilon)t}P_{\varepsilon}x(0)
\end{equation}
since the term $(I-P_{\varepsilon})x(0)$ decays exponentially to zero on the $\mathcal{O}(1)$ fast timescale, $t$. 

Of course, slow eigenspace projection is not the only way to resolve the multiscale dynamics of \eqref{linear}. Given that $A$ has distinct eigenvalues, special coordinates can be chosen to diagonalize $A$, thus {\it separating} the fast and slow dynamics. The slow {\it variable} is the variable $z$ for which
\begin{equation}\label{DIAG}
\begin{pmatrix}\dot{z}\\\dot{w}\end{pmatrix} = \begin{pmatrix}\lambda_+(\varepsilon) & 0\\ 0 & \lambda_-(\varepsilon)\end{pmatrix}\begin{pmatrix}z\\w\end{pmatrix}
\end{equation}
under the coordinate transformation $(z,w)^T=\mathcal{T}^{-1}\cdot (x_1,x_2)^T$, where the columns of $\mathcal{T}$ are the eigenvectors of $A$. Hence, a slow variable is any variable whose velocity vanishes as $\varepsilon \to 0$.\footnote{The more general definition of the slow variable and the {\it standard form} for nonlinear equations is presented in the Appendix.}

The aforementioned solution strategies have several shortcomings. First, slow eigenspace projection \eqref{linred} requires the direct computation of eigenvectors and eigenvalues, which can be exhaustive and lead to rather complicated expressions, even for planar systems. 
Second, if we separate the fast and slow dynamics by diagonalization \eqref{DIAG}, the new coordinates $(z,w)$ may not have a meaningful biochemical interpretation. This can be inconvenient given that our ultimate interest is in understanding when and how to reduce CMEs whose random variables represent integer-valued molecular copy numbers of specific chemical species.  
Third, and from a practical point of view, we are generally only interested in the leading-order approximation (in $\varepsilon)$ to solutions \eqref{linred} and \eqref{DIAG}. Although this is straightforward to compute
\begin{subequations}\label{reds1}
\begin{align}
z(\tau) &= e^{\displaystyle \lambda_+^{(1)}\tau}z(0),\\
x(\tau) &= e^{\displaystyle \lambda_+^{(1)}\tau}P_0x(0),
\end{align}
\end{subequations}
where $P_0$ is an appropriate zeroth-order approximation to $P_{\varepsilon}$ and $\tau = \varepsilon t$ is the slow timescale, it is desirable to compute such approximations without the need to deal with the eigenspectrum and eigenspaces of $A$ directly. 

Geometric singular perturbation theory provides a bridge from \eqref{linear} to \eqref{reds1} without the need to compute the eigenvectors and eigenvalues of $A$. To employ singular perturbation theory, we begin by setting $\varepsilon =0$ in \eqref{linear}, which generates the {\it layer problem}
\begin{equation}
\dot{x} = A_0x,
\end{equation}
where $A_0$ has a one-dimensional center subspace, $E^c$, corresponding to the eigenspace associated with the trivial eigenvalue, and a one-dimensional stable subspace, $E^s$, corresponding to the eigenspace associated with the nontrivial eigenvalue, $\lambda_-^{(0)}$. For the purposes of notation, let $v_+^{(0)}$ and $v_-^{(0)}$ satisfy
\begin{equation}
A_0v_+^{(0)}=0,\qquad A_0v_-^{(0)}= \lambda_-^{(0)}v_-^{(0)}, \quad \text{with}\;\;E^c = {\rm span}\{v_+^{(0)}\}\;\; \text{and} \;\;E^s={\rm span}\{v_-^{(0)}\}.
\end{equation}

\begin{definition}
Let $x=(x_1,x_2)$ and $w_0\in {\rm image} ^{\perp}A_0$ with $w_0^TA_0=0.$ The slow variable, $z$, is given by $z=\varphi(x_1,x_2)$ where
\begin{equation}\label{slowmap}
\varphi(x_1,x_2)=w_0^Tx.
\end{equation} 
\end{definition}

Any variable whose velocity does not vanish as $\varepsilon \to 0$ is called a {\it fast variable}. Note that the slow variable is identified with the first integral (conserved quantity) $w_0^Tx=const.$ of the layer problem. 

\begin{remark}\label{RM1}
The slow variable is not unique: any constant multiple of $z$ is also a slow variable. However, in the analysis that follows, we will define the slow variable uniquely in the following way. If $x_i$ is the parametric variable of interest (i.e., if we have parameterized $E^c$ as $x_2=x_2(x_1)$, then $x_1$ would be our parametric variable of interest), then we will define $\varphi(x)$ so that
\begin{equation}\label{21}
\cfrac{\partial \varphi}{\partial x_i}=1.
\end{equation}
This ensures that $x_i$ remains invariant with respect to the coordinate transformation that diagonalizes $A_0$; additional details are presented in the Appendix. 
\end{remark}

The center subspace $E^c$ is an example of a normally hyperbolic critical manifold $S_0$, which is the most fundamental (and important) feature of singularly perturbed differential equations.

\begin{definition}{\bf Normally hyperbolic and invariant submanifold}\footnote{This definition is somewhat restrictive and pertains {\it only} to the very special case of a submanifold of equilibrium points. Nevertheless, this rather narrow definition is sufficient for the purposes of this paper. See \citet{Eldering}, \citet{Fenichel1971}, and \citet{Wiggins} for a general and more technical treatment of the subject.} Consider the nonlinear ordinary differential equation
\begin{equation}
\dot{x} = h(x) + \varepsilon G(x,\varepsilon), \quad x\in \mathbb{R}^n.
\end{equation}
If the level set $S_0:=\{x\in \mathbb{R}^n: h(x)=0\}$ constitutes a $k$-dimensional submanifold of $\mathbb{R}^n$ (${\rm rank}\;Dh(x) = n-k \;\;\forall x\in S_0$), then $S_0$ is called the critical manifold and is normally hyperbolic if the spectrum of the Jacobian along $S_0$, $Dh(x)|_{x\in S_0}$, consists of a trivial eigenvalue with algebraic and geometric multiplicity $k$, and $n-k$ non-trivial eigenvalues that are strictly bounded away from the imaginary axis. Normal hyperbolicity ensures the splitting
\begin{equation*}
\mathbb{R}^n = \ker Dh(x) \oplus {\rm image} \;Dh(x)
\end{equation*}
holds for all $x\in S_0$, and asserts the existence of a unique projection matrix, $P_0:T_x\mathbb{R}^n \to T_xS_0\cong \ker Dh(x)$.
\end{definition}

Normally hyperbolic manifolds are structurally stable and persist -- along with their stable and unstable manifolds -- under smooth and small perturbations.\footnote{In this sense, they are the higher-dimensional analogs of hyperbolic equilibrium points.} Since the Jacobian $A_0$, along $E^c$, has a trivial eigenvalue with an algebraic and geometric multiplicity of one, and a single nontrivial eigenvalue, the manifold $E^c$ is normally hyperbolic. Turning on the perturbation by letting $0<\varepsilon$ results in the formation of a slow and invariant manifold, $S_{\varepsilon}$. For linear systems, the slow invariant manifold is simply the slow eigenspace of $A$.

Geometric singular perturbation theory procures the leading-order (in $\varepsilon$) approximation to \eqref{linred}, which is given by
\begin{equation}\label{Fenichel}
\dot{x} = \varepsilon P_0 A_1x|_{x\in E^c}, \qquad \text{with}\;\;P_0=\cfrac{v_+^{(0)}w_0^T}{w_0^Tv_+^{(0)}} = I -\cfrac{1}{\lambda_-^{(0)}}A_0,
\end{equation}
where $P_0:\mathbb{R}^2\to \ker A_0\equiv E^c$. If the center subspace is parameterizable by $x_1$, then its graph is given by
\begin{equation*}
E^c = \{(x_1,x_2)\in \mathbb{R}^2: x_2 = \ell  x_1\},
\end{equation*}
where $\ell$ is the slope. Therefore, in \eqref{Fenichel}, the restriction of $x$ to lie within $E^c$ yields
\begin{equation}\label{proj}
\dot{x} = \varepsilon P_0 A_1 x_1\begin{pmatrix}1\\ \cfrac{{\rm d}x_1}{{\rm d}x_2}\end{pmatrix} = \varepsilon P_0 A_1 x_1\begin{pmatrix}1\\ \ell\end{pmatrix} = \varepsilon P_0 A_1 x_1v_+^{(0)}
\end{equation}
where ${\rm d}x_1/{\rm d}x_2=\ell$ is the change in $x_1$ with respect to the change in $x_2$ along $E^c$. By the elementary perturbation theory of matrices (see \citet{greenbaum2019}), the projected system \eqref{proj} in terms of the parametric variable $x_1$ is
\begin{equation}\label{FR0}
\begin{pmatrix}x_1'\\x_2'\end{pmatrix} = \lambda_+^{(1)}x_1\begin{pmatrix}1\\\ell \end{pmatrix},
\end{equation}
where ``$\phantom{x}'$" denotes differentiation with respect to the slow timescale, $\tau$.

The reduced equation \eqref{FR0} is called the {\it Fenichel reduction} of \eqref{linear}, and it provides an accurate approximation of the dynamics in the slow eigenspace of $A$ that is valid as $\varepsilon \to 0$. In fact, the solution to \eqref{FR0} is the leading order approximation to \eqref{linred} in $\varepsilon$. Note that as long as $\ell \neq 0$, we can alternatively parameterize $E^c$ by $x_2$:
\begin{equation*}
\begin{pmatrix}x_1'\\x_2'\end{pmatrix} = \lambda_+^{(1)}x_2\begin{pmatrix}\ell^{-1}\\1\end{pmatrix}.
\end{equation*}

In biochemistry, Fenichel reductions of the form \eqref{FR0} are called quasi-steady-state approximations (QSSAs) and are commonly used in the analysis of mass action equations that model reaction networks~\cite{Segel1988,Segel1989,Schnell1997,Heineken1967,Wilhelm2000,Cracium}.  The basic idea in GSPT is that, since the stable and unstable manifolds of the critical manifold also persist under sufficiently smooth and small perturbations, the resulting normally hyperbolic invariant manifold (again, the slow eigenspace for linear problems) inherits the stability properties of the critical manifold (the center subspace, $E^c$). Moreover, normal hyperbolicity ensures the existence of a unique $P_0$, and thus the long-time dynamics can be approximated by projecting the perturbation onto $TE^c$, the tangent bundle of $E^c$, given by
\begin{equation*}
TE^c = \bigcup_{x\in E^c, v\in E^c} (x,v) = E^c \times E^c.
\end{equation*}
The construction of the projection matrix $P_0$ for general nonlinear systems is discussed in~\citet{Fenichel1979} and~\citet{Wechselberger2020}. Although the following is encoded in both~\cite{Fenichel1979,Wechselberger2020}, it will be useful to express $P_0$ in terms of derivatives.
\begin{proposition}
For a linear planar system, the projection matrix, $P_0$, has the following representation:
\begin{equation}\label{pmder}
P_0 = \begin{pmatrix}\cfrac{{\rm d}x_1}{{\rm d}z}\\\cfrac{{\rm d}x_2}{{\rm d}z}\end{pmatrix}\begin{pmatrix}\cfrac{\partial \varphi}{\partial x_1}& \cfrac{\partial \varphi}{\partial x_2}\end{pmatrix},
\end{equation}
where $z$ is the slow variable, $\varphi(x_1,x_2)$ is given by \eqref{slowmap} and ${\rm d}x_1/{\rm d}z, {\rm d}x_2/{\rm d}z$ are the slopes of the lines $x_1=\ell_1z$ and $x_2 = \ell_2 z$ that define $E^c$ in $(x_1,z)$ and $(x_2,z)$ coordinates.
\end{proposition}
\begin{proof}
It is straightforward to verify $P_0^2=P_0$ using the identity
\begin{equation*}
\cfrac{{\rm d}x_1}{{\rm d}z}\cdot\cfrac{\partial \varphi}{\partial x_1}+\cfrac{{\rm d}x_2}{{\rm d}z}\cdot\cfrac{\partial \varphi}{\partial x_2} =1.
\end{equation*}
Projection matrices are uniquely defined by their kernel and image. The span of the vector $\begin{pmatrix}\cfrac{\partial \varphi}{\partial x_1}& \cfrac{\partial \varphi}{\partial x_2}\end{pmatrix}^T$ defines the kernel of $A_0^T$, which is the orthogonal complement to the image of $A_0$. Thus,
\begin{equation}\label{ker1}
\ker P_0 = {\rm image} \;A_0 = {\rm span}\{v_-^{(0)}\}.
\end{equation}
Next, suppose $E^c$ is parameterizable by $x_1$ in which case $v_+^{(0)}$ has the representation
\begin{equation*}
v_+^{(0)} = \begin{pmatrix}1\\\cfrac{{\rm d}x_2}{{\rm d}x_1}\end{pmatrix}.
\end{equation*}
By definition, ${\rm d}z/{\rm d}x_1= w_0^Tv_+^{(0)}.$ Moreover, normal hyperbolicity of $E^c$ ensures $w_0^Tv_+^{(0)}\neq 0$ (see {\sc figure} \ref{FIG3}), and therefore 
\begin{equation*}
(w_0^Tv_+^{(0)})^{-1}\begin{pmatrix}1\\\cfrac{{\rm d}x_2}{{\rm d}x_1}\end{pmatrix} = \begin{pmatrix}\cfrac{{\rm d}x_1}{{\rm d}z}\\\cfrac{{\rm d}x_2}{{\rm d}z}\end{pmatrix}
\end{equation*}
also belongs to the kernel of $A_0$. The equivalent argument holds when $E^c$ is parameterizable by $x_2$. Thus, {\rm image} of $P_0 = \ker A_0$.
\end{proof}
\begin{figure}[bth!]
    \centering
    \includegraphics[scale=1.0]{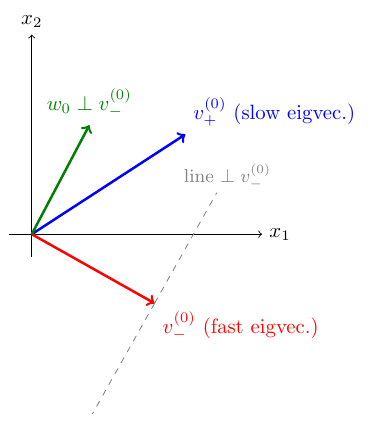}
    \caption{\textbf{The geometry of normal hyperbolicity:} $w_0$ is orthogonal to the fast eigenvector. Thus, if $w_0^Tv_+^{(0)}=0$, then $v_+^{(0)} \parallel v_-^{(0)}$. Consequently, the algebraic multiplicity of the trivial eigenvalue is $2$, but the geometric multiplicity is $1$. If $w_0^Tv_+^{(0)}=0$, then $E^c$ is not normally hyperbolic since by definition the geometric and algebraic multiplicities of the trivial eigenvalue must both equal $1$.}
    \label{FIG3}
\end{figure}

To summarize, geometric singular perturbation theory provides a methodology to obtain the leading order approximation of the dynamics on the slow eigenspace. As an approximation, the resulting expressions are often simpler than those of the exact solution. Moreover, the Fenichel reduction is constructed from the left and right eigenvectors of $A_0$, thus circumnavigating the need to deal directly with $A$. 

\section{The linear noise approximation}\label{section:3}

The linear noise approximation, originally derived by Van~\citet{VKX}, provides an important bridge between the deterministic mass action equations and the stochastic chemical master equation.\footnote{The rigorous mathematical justification of the LNA was established by~\citet{kurtz1978}.} In this section, we will introduce the LNA and examine its reduction using Fenichel theory. Although we take a coordinate-independent approach to the reduction of stochastic LNAs, geometric singular perturbation and center manifold theory have been rigorously applied to SDEs; see~\cite{Yin2004,Yin2005,BerglundGentzJDE,KUEHNstoch,Roberts1996,Roberts2013,KnoblochCM} for more technical analyses, and \cite{PAHLAJANI201196,Popovic} for analyses concerning stochastic reaction networks with disparate timescales.

In their most compact form, mass action equations can be written as
\begin{equation}\label{MAC}
\dot{x}_i = S_{ij}r_j(x),
\end{equation}
where $S_{ij}$ is a net stoichiometric matrix and $r_j({x})$ are the elementary reactions that comprise the network. For example, consider the following reaction network:
\begin{align}\label{GC}
 \ce{A <=>[$k_1$][$k_{-1}$] 2B},\quad \ce{B ->[$k_2$] C}
\end{align}
The mass action equations corresponding to \eqref{GC}
\begin{subequations}
\begin{align*}
\dot{a} &= -k_1a + k_{-1}b^2,\\
\dot{b} &= \;2k_1a -2k_{-1}b^2 -k_2b,
\end{align*}
\end{subequations}
which admits the alternative form
\begin{equation}
\begin{pmatrix}\dot{a}\\\dot{b}\end{pmatrix} = \begin{pmatrix}-1 & \;\;1 & 0\\ \;\;2 & -2 & 1\end{pmatrix}\begin{pmatrix}k_1a\\k_{-1}b^2\\k_2b\end{pmatrix}.
\end{equation}

In the linear noise regime, the first moment is the solution to the mass action system. Thus, the mean field concentrations are simply the macroscopic solutions, $x_i(t)$, to the mass action equations.
Near the system size limit, the concentrations will fluctuate randomly about the mean field concentration. These fluctuations, denoted by $X_i$, satisfy the linear diffusion process,
\begin{equation}\label{LNA}
{\rm d}{X} = D(S{r}({x}))X \;{\rm d}t + \Omega^{-1/2}B({x}){\rm d}{W}(t)
\end{equation}
where $X=(X_1,X_2,...X_n)^T$, $D$ denotes differentiation with respect to mean field concentrations, $x_i$, and ${\rm d}{W}$ denotes the column vector $({\rm d}W_1,{\rm d}W_2,...{\rm d}W_k)^T$ for a network comprised of ``$k$" elementary reactions and ``$n$" chemical species where $W_i(t)$ are standard Brownian motions:
\begin{subequations}
\begin{align}
\mathbb{E}\{W_i(t)\} &= 0,\\
\mathbb{E}\{W_i(t)W_i(s)\} &= \min\{t,s\}. 
\end{align}
\end{subequations}
The matrix $B$ is given by
\begin{equation}\label{BMatrix}
B({x}) = S\left({\rm diag}\left(\sqrt{{r}_j({x})}\right)\right).
\end{equation}

The second moment (covariance) is a matrix quantity,
\begin{equation}
C_{ij} = {\rm Cov}(X_i,X_j),
\end{equation}
which satisfies the Lyapunov matrix differential equation
\begin{equation}\label{lyapunov}
\dot{C} = J(x)C + CJ(x)^T + \Omega^{-1}B({x})B({x})^T, \quad J := D(S{r({x}})).
\end{equation}

\begin{remark}
For the purposes of the following analysis, it suffices to set $\Omega =1$ from this point on. However, it should be understood that when we introduce $E^c$ in the context of fluctuations $X_i$, this is the center subspace that emerges in both the singular limit $\varepsilon \to 0$ and the system size limit: $\Omega \to \infty$. Thus, derivatives of the form
\begin{equation*}
\cfrac{{\rm d}X_i}{{\rm d}Z}=\ell,
\end{equation*}
are to be interpreted in the singular and thermodynamic (system size) limits. Note that for linear systems ${\rm d}x_i/{\rm d}z={\rm d}X_i/{\rm d}Z$.
\end{remark}

Since the first and second moments satisfy ordinary differential equations, GSPT can be applied directly to produce reduced equations that approximate the long-time behavior of moments on the slow timescale $\tau$. Since we are interested in first-order reaction networks whose mass action equations are of the form \eqref{linear}, we are interested in reducing LNAs of the form
\begin{subequations}
\begin{align}
\dot{x} &= (A_0+\varepsilon A_1)x,\label{MF}\\
{\rm d}X &= (A_0+\varepsilon A_1)X \;{\rm d}t + B(x,\sqrt\varepsilon){\rm d}W(t).\label{OU}
\end{align}
\end{subequations}
For linear mass action networks, the reduction of the  SDE \eqref{OU} is given by
\begin{equation}
    {\rm d}X = \varepsilon \lambda_+^{(1)}X_1\begin{pmatrix}1 \\ \ell \end{pmatrix}{\rm d}t + P_0B(x_1,\sqrt{\varepsilon}){\rm d}W(t),
\end{equation}
where we have symbolically chosen to parameterize $E^c\times E^c$ by $x_1$ and $X_1$. The matrix $B(x,\sqrt{\varepsilon})$ is given by
\begin{equation}
    B(x,\sqrt{\varepsilon}) = S_0\left( {\rm diag} \sqrt{r_{1\leq j\leq p}(x)}\right) + \sqrt{\varepsilon}S_1\left( {\rm diag} \sqrt{r_{p< j\leq k}(x)}\right) =: B_0(x) + \sqrt{\varepsilon}B_1(x),
\end{equation}
where $S_0$ is the net stoichiometric matrix corresponding to the ``$p$" fast $\mathcal{O}(1)$ reactions, and $S_1$ is the net stoichiometric matrix corresponding to the ``$k-p$" slow $\mathcal{O}(\varepsilon)$ reactions. Since the columns of $S_0$ lie in the column space of $A_0$, $P_0S_0$ vanishes, and we are left with
\begin{equation}\label{rSDE}
{\rm d}X =  \lambda_+^{(1)}X_1\begin{pmatrix}1 \\ \ell\end{pmatrix}{\rm d}\tau + P_0B_1(x_1){\rm d}W(\tau),\quad \tau = \varepsilon t,
\end{equation}
where $\ell$ denotes the slope of $E^c$
\begin{equation*}
   X_2=\ell X_1, \quad \ell = \cfrac{{\rm d}x_2}{{\rm d}x_1},
\end{equation*}
the center subspace (critical manifold) of $A_0$.

The SDE \eqref{rSDE} corresponds to the Fenichel reduction of the covariance equation\footnote{The critical manifold associated with the covariance equation is defined by the set $\{M\in S_2(\mathbb{R}): \mathcal{L}_{A_0}(M) = -B_0(x)B_0(x)^T\}$, where $S_2(\mathbb{R})$ is the set of real symmetric $2\times 2$ matrices, however, due to the linearity we can just restrict $C\in \ker \mathcal{L}_{A_0}$ to obtain the Fenichel reduction. See the example presented in the Appendix for further details.},
\begin{equation}\label{Cprime}
C' = P_0\mathcal{L}_{A_1}(C)P_0^T + P_0B_1(x_1)B_1(x_1)^TP_0^T, \quad C\in \ker \mathcal{L}_{A_0}(\cdot)
\end{equation}
where again ``$\phantom{x}'$" denotes differentiation with respect to slow time, $\tau$, and $\mathcal{L}_A(\cdot)$ denotes the Lyapunov operator: $\mathcal{L}_A(\cdot)=A(\cdot)+(\cdot)A^T$.

Analyzing \eqref{Cprime} and \eqref{rSDE} directly is sufficient for most applications, but these projected forms can obfuscate certain dynamic relationships. For example, looking closely at \eqref{Cprime}, we obtain the following:
\begin{proposition}\label{prop2}
For a linear reaction network, the Fenichel reduction of the covariance matrix, $C$, is given by
\begin{equation}\label{FC}
C'= (2\lambda_+^{(1)}{\rm Var}(Z)(\tau) +|\lambda_+^{(1)}|z(\tau))\begin{pmatrix}\cfrac{{\rm d}x_1}{{\rm d}z}\cfrac{{\rm d}x_1}{{\rm d}z} & \cfrac{{\rm d}x_1}{{\rm d}z}\cfrac{{\rm d}x_2}{{\rm d}z} \\ \cfrac{{\rm d}x_1}{{\rm d}z}\cfrac{{\rm d}x_2}{{\rm d}z} & \cfrac{{\rm d}x_2}{{\rm d}z}\cfrac{{\rm d}x_2}{{\rm d}z}\end{pmatrix}
\end{equation}
where ${\rm Var}(Z(t))=\mathbb{E}(Z^2)-[\mathbb{E}(Z)]^2$, $Z$ is a diffusion process defined by
\begin{equation}
{\rm d}Z = \lambda_+^{(1)}Z\;{\rm d}\tau-\sqrt{|\lambda_+^{(1)}|z}\;{\rm d}W(\tau).
\end{equation}
The variance, ${\rm Var} (Z)$, satisfies the ordinary differential equation
\begin{equation}
{\rm Var}(Z)' = 2\lambda_+^{(1)}{\rm Var}(Z)+|\lambda_+^{(1)}|z,
\end{equation}
where again lowercase $z$ denotes the deterministic slow variable that satisfies the evolution equation $z'=\lambda_+^{(1)}z$.
\end{proposition}
\begin{proof}
This is an elementary computation involving \eqref{pmder}: differentiate $P_0CP_0^T$ with respect to time and keep only the leading order terms in $\varepsilon$. See the Appendix for additional details.
\end{proof}

The geometric information contained in Proposition \ref{prop2} should be emphasized. What Proposition \ref{prop2} ultimately says is the following: Any variable $x_i$ of a first-order reaction network whose mass action equations are singularly perturbed admits a reduced LNA on the slow timescale, $\tau$, of the form
\begin{equation}\label{VF}
{\rm d}X_i = \lambda_+^{(1)}X_i\;{\rm d}\tau - \sqrt{|\lambda_+^{(1)}|x_i\left(\cfrac{{\rm d}x_i}{{\rm d}z}\right)}\;{\rm d}W(\tau).
\end{equation}
Most importantly, we have the following from \eqref{VF}:
\begin{proposition}\label{prop3}
Let $(x_1,x_2)$ denote the concentrations of chemical species that comprise a first-order reaction network that is singularly perturbed and with a leading order Jacobian, $A_0$, with an eigenspectrum consisting of a single trivial eigenvalue and one eigenvalue that is strictly negative and real. If ${\rm d}x_i/{\rm d}z=1$, where $z$ is the slow variable given by $z=\varphi(x_1,x_2)$ with $\partial_{x_1}\varphi =1$, then the corresponding CME of the network admits a reduction in the form of a Poisson death process
\begin{equation}\label{DP}
\cfrac{\partial}{\partial t}Pr(n,t) = \varepsilon |\lambda_+^{(1)}|(E_n^{-1}-1)nPr(n,t)
\end{equation}
where $n$ denotes molecular copy numbers of the chemical species associated with $x_i$.
\end{proposition}
\begin{proof}
Since ${\rm d}x_i/{\rm d}z=1$, the reduced LNA for $x_i$ is
\begin{subequations}\label{LLN}
\begin{align}
x_i' &= \lambda_+^{(1)}x_i,\\
{\rm d}X_i &= \lambda_+^{(1)}X_i\;{\rm d}\tau - \sqrt{|\lambda_+^{(1)}|x_i}\;{\rm d}W(\tau)
\end{align}
\end{subequations}
by Proposition \ref{prop2}. The system size expansion of \eqref{DP} (see \cite{VKX} for details) is exactly \eqref{LLN}, and the assertion follows. 
\end{proof}

Proposition \ref{prop3} reveals that the slow variable, $z$, is an organizing agent in the sense that {\it all} reduced LNAs (for each $x_i)$ can be derived from the LNA of the slow variable, provided that the slopes (${\rm d}x_i/{\rm d}z$) can be easily calculated. Therefore, first order singularly perturbed two-dimensional reaction networks that depend on a small parameter $\varepsilon$, admit a reduced CME whose propensity function is the stochastic analog of the deterministic reduction as long as the parametric variable of interest $x_1$, satisfies
\begin{equation}\label{GEOM}
\cfrac{{\rm d}z}{{\rm d}x_1} = \cfrac{\partial \varphi(x_1,x_2)}{\partial x_1} + \cfrac{{\rm d}x_2}{{\rm d}x_1} \cdot\cfrac{\partial \varphi(x_1,x_2)}{\partial x_2} =  1+\cfrac{{\rm d}x_2}{{\rm d}x_1} \cdot\cfrac{\partial \varphi(x_1,x_2)}{\partial x_2}=w_0^Tv_+^{(0)}=1.
\end{equation}
When \eqref{GEOM} holds, then $Pr(n,t)$, the probability that there are ``$n$" molecules of species $x_i$ at time $\tau$, is well approximated by
\begin{equation}
Pr(n,\tau) = \begin{pmatrix}N\\n\end{pmatrix}e^{-\displaystyle n\lambda_+^{(1)}\tau}\left(1-e^{-\displaystyle \lambda_+^{(1)}\tau}\right)^{(N-n)}
\end{equation}
where $\tau$ is the slow timescale, $\tau=\varepsilon t$, $N$ is the total number of molecules corresponding to species $x_1$, and $\varepsilon \lambda_+^{(1)}$ is the first-order approximation to the slow eigenvalue, $\lambda_+(\varepsilon)$. Looking carefully at \eqref{GEOM}, it is clear that the second term must vanish
\begin{equation}\label{vanished}
\cfrac{{\rm d}x_2}{{\rm d}x_1} \cdot\cfrac{\partial \varphi(x_1,x_2)}{\partial x_2} =0
\end{equation}
if a reduction of the CME based on $x_1$ is possible. If ${\rm d}x_2/{\rm d}x_1=0$, then $E^c$ corresponds to the $x_1$ coordinate axis. On the other hand, if $\partial_{x_2}\varphi(x_1,x_2)=0$, then $x_1$ {\it is} the slow variable: $z=x_1$. If both hold, then $A_0$ is diagonal
\begin{equation*}
A_0 = \begin{pmatrix}0&0\\0&\lambda_-^{(0)}\end{pmatrix},
\end{equation*}
and the projection matrix, $P_0$, is symmetric.

An important consequence of \eqref{vanished} is the following:
\begin{proposition}\label{prop6}
Suppose 
\begin{equation*}
\cfrac{{\rm d}x_2}{{\rm d}x_1} \cdot\cfrac{\partial \varphi(x)}{\partial x_2} =0,\quad \text{with} \;\;\cfrac{\partial \varphi(x)}{\partial x_1}=1
\end{equation*}
for a first order reaction network. Then, the diffusion of $X_1$ is limited to the slow timescale, $\tau$, whenever $x=(x_1,x_2)$ is restricted to evolve on the slow manifold.
\end{proposition}\label{prop5}
\begin{proof}
If $\partial_{x_2}\varphi(x)=0$, then $x_1$ {\it is} the slow variable and therefore, by linearity, $x_1=z$ and $X_1=Z$. By definition of the slow variable, both the drift and diffusion of $Z$ occur solely on the slow timescale.
On the other hand, if $\partial_{x_2}\varphi(x)\neq 0$ but ${\rm d}x_2/{\rm d}x_1=0$, then $E^c$ is the $x_1$-axis and the layer problem admits the form
\begin{equation*}
\dot{x} = A_0x \equiv S_0r(x_2),\;\;\text{with}\;\;r_i(x_2)=k_ix_2
\end{equation*}
since the critical manifold coincides with $x_2=0$. Hence, all fast ($\mathcal{O}(1)$) reactions have the form $r_jx_2$, and all slow ($\mathcal{O}(\varepsilon)$) reactions have the form $\varepsilon r_ix_1$. The asymptotic expansion of the slow eigenspace has the form
\begin{equation}\label{Aexpan}
x_2 = \varepsilon \xi_1(x_1)+ \varepsilon^2 \xi_2(x_1) + \mathcal{O}(\varepsilon^3),
\end{equation}
where $\xi_i(x_1)$ are linear in $x_1$. When the expression \eqref{Aexpan} is substituted in $B(x)$ as defined by \eqref{BMatrix}, every non-zero term in $B(x)$ is $\mathcal{O}(\sqrt{\varepsilon})$, and the assertion follows.
\end{proof}

\begin{remark}
Proposition \ref{prop5} is somewhat counterintuitive. It seems reasonable that if the initial conditions of the mass action system are selected to lie within the slow eigenspace so that
\begin{equation*}
\frac{{\rm d}x}{{\rm d}t} \sim \mathcal{O}(\varepsilon),\;\;t\geq 0
\end{equation*}
then the fluctuations, $X$, about the mean field, $x$, will also be ``slow" in the sense that the variance, ${\rm Var}(X)$, should remain effectively frozen (constant) over fast timescales:
\begin{equation*}
\cfrac{{\rm d}}{{\rm d}t}{\rm Var}(X)\sim \mathcal{O}(\varepsilon), \;\;t\geq 0.
\end{equation*}
However, this will not be the case unless \eqref{vanished} is satisfied. When the critical manifold does not correspond to a coordinate axis, the asymptotic expansion of the slow eigenspace is of the form
\begin{equation}\label{Aexpan2}
x_2 = \xi_0(x_1) + \varepsilon \xi_1(x_1)+ \varepsilon^2 \xi_2(x_1) + \mathcal{O}(\varepsilon^3),
\end{equation}
and therefore the leading order term is $\mathcal{O}(1)$ instead of $\mathcal{O}(\varepsilon)$ as it is in \eqref{Aexpan}. Consequently, the substitution of \eqref{Aexpan2} in $B(x)$ will result in terms that are $\mathcal{O}(1)$ and $\mathcal{O}(\sqrt{\varepsilon})$; thus, diffusion will occur on both fast and slow timescales, even though the initial conditions assigned to the mass action system lie in -- or sufficiently close to -- the slow eigenspace. A concrete example of this statement is presented in the Appendix, and is also thoroughly discussed in~\citet{JEWS2026}.
\end{remark}

\section{First order reaction networks and the heuristic reduction of the CME}\label{section:4}

In this section, we invoke Proposition \ref{prop3} and revisit the reductions introduced in Section 1. Specifically, we will work out the slow timescale reductions of the LNA from the linear reaction network \eqref{reac1} for both small $k_1$ and small $k_3$, and demonstrate why the heuristically-reduced CME is accurate for the former case but not for the latter. It will be helpful to recall It\^{o}'s formula, taken directly from~\citet{kuehn2015}:
\begin{definition}\label{Ito}
{\bf{It\^{o}'s Formula}}: Let $z$ satisfy the diffusion process
\begin{equation*}
{\rm d}z = a(z)\;{\rm d}t + \sigma A(z)\;{\rm d}W(t)
\end{equation*}
and suppose $v=u(z,t)$. Then, $v$ satisfies
\begin{equation*}
{\rm d}v= \left[\cfrac{\partial u}{\partial t}+\cfrac{\partial u}{\partial z}a(z) + \cfrac{1}{2}\cfrac{\partial^2u}{\partial z^2}A(z)^2\right]{\rm d}t + \cfrac{\partial u}{\partial z}A(z){\rm d}W(t).
\end{equation*}
\end{definition}
The examples presented in this section are not only consistent with It\^{o}'s formula, but are in some sense interpretable as applications of It\^{o}'s formula. In essence, Propositions 1-3 highlight the following: for linear reaction networks, reductions on the slow time scale, $\tau$, are determined by linear maps $x_i=\ell_i z$ and $X_i=\ell_i Z$, where $z$ and $Z$ are the respective slow variables. If the slow variable satisfies the diffusion process,
\begin{equation*}
{\rm d}Z= \lambda Z\;{\rm d}\tau -\sqrt{|\lambda| z}\;{\rm d}W(\tau),\quad \lambda < 0
\end{equation*}
then, in accordance with It\^{o}'s formula, the slow drift and diffusion of $X_i$ satisfy
\begin{equation*}
{\rm d}X= \ell\lambda Z\;{\rm d}\tau -\sqrt{\ell^2|\lambda| z}\;{\rm d}W(\tau) = \lambda X\;{\rm d}\tau -\sqrt{\ell|\lambda| x}\;{\rm d}W(\tau),
\end{equation*}
since $x=\ell z$ and $X=\ell Z$. The significance of this observation is that a Poisson process is defined by a master equation of the form
\begin{equation}\label{DP1}
\cfrac{\partial}{\partial t}Pr(n,t) = \varepsilon |\lambda|(E_n^{-1}-1)nPr(n,t), \quad \lambda <0
\end{equation}
has an associated LNA given by 
\begin{subequations}\label{RF}
\begin{align}
x' &=  \lambda x,\\
{\rm d}X&= \lambda X\;{\rm d}\tau -\sqrt{|\lambda| x}\;{\rm d}W(\tau)
\end{align}
\end{subequations}
where $\lambda <0$. Thus, our objective will be to show that the Fenichel reduction of the LNA does not assume the form \eqref{RF} when $\ell \neq 1$, and therefore the heuristically reduced CME will not approximate the variance of $X$.

\begin{example}
Consider \eqref{reac1} with small $k_1$. The mass action equations are 
\begin{subequations}\label{sys1RV}
\begin{align}
\dot{x} &= -\varepsilon k_1 x + k_2y,\label{xdotRV}\\
\dot{y} &= \;\;\;\varepsilon k_1x - k_2y - k_3y.\label{ydotRV}
\end{align}
\end{subequations}
The Jacobian of the layer problem is
\begin{equation}\label{layer1}
A_0=\begin{pmatrix}
0 & \;\;\;k_2\\0 & -k_2-k_3
\end{pmatrix}, \quad \text{with}\;\; v_+^{(0)} = \begin{pmatrix}1\\0\end{pmatrix}, \quad w_0=\begin{pmatrix}1\\ \cfrac{k_2}{k_2+k_3}\end{pmatrix}.
\end{equation}
In this case, the center subspace $E^c$ is the $x$-axis and is therefore not parameterizable by $y$. Thus, we will look for a reduced equation for $x$ (hence, $x$ is the parametric variable of interest). From \eqref{layer1} we have the following:
\begin{equation}\label{EIGS1}
    \cfrac{{\rm d}x}{{\rm d}z} = (w_0^Tv_+^{(0)})^{-1} =1, \qquad \varepsilon \lambda_+^{(1)} = \cfrac{w_0^TA_1v_+^{(0)}}{w_0^Tv_+^{(0)}}= -\cfrac{\varepsilon k_1 k_3}{k_2+k_3}.
\end{equation}
Thus, from \eqref{EIGS1} and Proposition \ref{prop3}, the CME of the reaction network \eqref{reac1} with small $k_1$ admits the reduced CME given by \eqref{rCME1}. 
\end{example}

\begin{example}
We now consider the reaction network \eqref{reac1} with small $k_3$. The mass action equations in this case are
\begin{subequations}\label{sys1RV1}
\begin{align}
\dot{x} &= -k_1 x + k_2y,\label{xdotRV1}\\
\dot{y} &= \;\;\;k_1x - k_2y - \varepsilon k_3y,\label{ydotRV1}
\end{align}
\end{subequations}
and by inspection we have:
\begin{equation*}
A_0 = \begin{pmatrix}-k_1 & \;\;\;k_2\\\;\;\;k_1& -k_2\end{pmatrix}.
\end{equation*}
If we choose to parameterize $E^c$ by $x$, then we can express $v_+^{(0)}$ as
\begin{equation}\label{vec}
v_+^{(0)}= \begin{pmatrix}1 \\ \cfrac{k_1}{k_2}\end{pmatrix} \quad \text{with}\;\;w_0=\begin{pmatrix}1\\1\end{pmatrix}.
\end{equation}
From \eqref{vec}, we obtain 
\begin{equation}\label{58}
\cfrac{{\rm d}x}{{\rm d}z}= \cfrac{k_2}{k_1+k_2},\qquad \varepsilon \lambda_+^{(1)} = -\cfrac{\varepsilon k_3 k_1}{k_1+k_2},
\end{equation}
and therefore the projected (reduced) LNA of $x$ on the slow timescale is
\begin{subequations}\label{LN00}
\begin{align}
x' &= -\cfrac{k_3k_1}{k_1+k_2}\cdot x,\label{H1}\\
{\rm d}X &= -\cfrac{k_3k_1}{k_1+k_2}\cdot X\;{\rm d}\tau- \sqrt{\cfrac{k_3k_1}{k_1+k_2}\cdot\textcolor{red}{\left(\cfrac{k_2}{k_1+k_2}\right)}\cdot x}\;{\rm d}W(\tau)\label{LN1}
\end{align}
\end{subequations}
and therefore the heuristic CME based on \eqref{H1} given by \eqref{rCME2} will fail to accurately approximate the variance. The same argument holds for $y$, for which we have
\begin{subequations}\label{LN0}
\begin{align}
y' &= -\cfrac{k_3k_1}{k_1+k_2}\cdot y,\label{H1y}\\
{\rm d}Y &= -\cfrac{k_3k_1}{k_1+k_2}\cdot Y\;{\rm d}\tau- \sqrt{\cfrac{k_3k_1}{k_1+k_2}\cdot\textcolor{red}{\left(\cfrac{k_1}{k_1+k_2}\right)}\cdot y}\;{\rm d}W(\tau).\label{LN2}
\end{align}
\end{subequations}
However, from \eqref{vec} it is clear that the slow variable is the sum $z=x+y$, which is dynamically equivalent to $p$, the concentration of $P$. Thus, the reduced CME \eqref{rCME3} accurately approximates the stochastic formation of $P$ molecules (again, see {\sc figure} \ref{FIG2} for numerical confirmation). 
\end{example}

A complete analysis (over both fast and slow timescales) of Example 2 is presented in the Appendix. Propositions 1--3, as well as examples 1--2, reveal that the heuristically reduced CME based on the deterministic reductions is accurate for linear (first order) reaction networks, as long as ${\rm d}x_i/{\rm d}z=1$, where $x_i$ is the parametric variable (concentration) of interest. It should be noted that Proposition \ref{prop2} is consistent with It\^{o}'s formula \ref{Ito}: For first order reaction networks, the maps $v=u(z,t)$ are linear and have no explicit time dependence: $v=\ell z$. In the nonlinear regime (i.e., when the reaction network is nonlinear), It\^{o}'s formula ensures nonlinear coordinate transformations will generate additional (noise-induced) drift terms, which appear due to biases introduced by the curvature of the slow manifold as well as the curvature of the flow field; see~\citet{Parsons2017}, as well as~\citet{Katz}. The takeaway from this is that a heuristically reduced CME is likely to be inaccurate whenever the coordinate transformations between variables of interest are nonlinear. We will explore this idea further in Section \ref{section:5}. 

\section{Nonlinear reaction networks}\label{section:5}

The implication of our linear analysis raises two obvious questions concerning the heuristic reduction of the master equation when the reaction network is nonlinear. First, are heuristically reduced CMEs accurate near mean field stationary points where the mass action equations can be approximated linearly? Secondly, are such reductions accurate when the coordinate transformation to the slow variable is linear despite the nonlinearity of the reaction mechanism? We will examine the answers to these questions in this section by analyzing the nonlinear Michaelis-Menten reaction network, as well as \eqref{GC}, which was previously introduced in Section \ref{section:3}.

\subsection{The Michaelis-Menten reaction network and the stochastic QSSA}
The irreversible Michaelis-Menten reaction network
\begin{align}\label{mm1}
 \ce{S + E <=>[$k_1$][$k_{-1}$] C ->[$k_2$] E + P},
\end{align}
describes the catalysis of a substrate, $S$, into a product, $P$ by an enzyme, $E$. When the concentration of enzyme is small, the network is modeled by the following nonlinear mass action system
\begin{subequations}\label{mmMASS}
\begin{align}
\dot{s} &= -k_1(\varepsilon e_0-c)s + k_{-1}c,\\
\dot{c} &= \;\;k_1(\varepsilon e_0-c)s -k_{-1}c-k_2c,
\end{align}
\end{subequations}
where $s$ and $c$ are the concentrations, $S$, and the complex, $C$, and $e_0$ is the total enzyme (bound or unbound), $E$, which is a conserved quantity. 

The mass action system \eqref{mmMASS} is a singularly perturbed differential equation. On the slow timescale, the Fenichel reduction of \eqref{mmMASS} is the well-studied quasi-steady-state approximation (QSSA),
\begin{subequations}\label{QSSA}
\begin{align}
s' &= -\cfrac{e_0k_2s}{s+K_M},\qquad K_M:=\cfrac{k_{-1}+k_2}{k_1}\\
c'&=0,
\end{align}
\end{subequations}
is valid whenever the total enzyme concentration is much lower than the Michaelis constant: $e_0\ll K_M$; see~\cite{Goeke2012,Goeke2013,Goeke2015} for rigorous derivations of QSSA's for enzymatic reaction networks. 

In~\cite{Rao2003}, Rao and Arkin introduced what is now termed the {\it stochastic} QSSA 
\begin{equation}\label{stochMM}
\cfrac{\partial}{\partial t}Pr(n,t) = (E_n^{+1}-1)\cfrac{e_0k_2n}{n+K_M}Pr(n,t)
\end{equation}
where again, $e_0$ is the total enzyme concentration and $n$ is the number of substrate (S) molecules. Based on scaling analysis, Rao and Arkin argued that the reduction of the CME based on the QSSA \eqref{mmMASS} is valid under the same conditions that support the validity of the deterministic QSSA. However, several years later~\citet{Sanft} reported that the stochastic QSSA \eqref{stochMM} appeared to occasionally {\it overestimate} the variance of $n$ throughout the stochastic timecourse. The observation published by~\citet{Sanft} was based on qualitative assessments of simulations. Moreover, \citet{Thomas2011} reported that the stochastic QSSA \eqref{stochMM}, when applied to an open Michaelis-Menten reaction network with a constant influx of substrate, overestimated the steady-state variance. In followup studies, \citet{Thomas2012,ThomasPO} introduced the {\it slow scale} LNA, and found that for substrate, $s$, the mean field fluctuations, $X_s$, satisfy
\begin{equation}\label{ssLNA}
{\rm d}X_s = -\cfrac{\partial}{\partial s}\left(\cfrac{k_2e_0s}{s+K_M}\right)X_s\;{\rm d}\tau -\sqrt{\cfrac{k_2e_0s}{s+K_M}\left(1-\cfrac{2K\cdot s}{(s+K_M)^2}\right)}\;{\rm d}W(\tau),\quad K:=k_2/k_1
\end{equation}
which is noticeably different from the LNA corresponding to \eqref{stochMM}, given by \eqref{SLNA}:
\begin{equation}\label{SLNA}
{\rm d}X_s = -\cfrac{\partial}{\partial s}\left(\cfrac{k_2e_0s}{s+K_M}\right)X_s\;{\rm d}\tau -\sqrt{\cfrac{k_2e_0s}{s+K_M}}\;{\rm d}W(\tau).
\end{equation}
Comparing \eqref{ssLNA} with \eqref{SLNA}, it is clear that the heuristically reduced CME \eqref{stochMM} only provides an accurate steady-state substrate variance estimate whenever 
\begin{equation}
\eta(s) :=\cfrac{2K\cdot s}{(s+K_M)^2} \ll 1.
\end{equation}

The function $\eta(s)$ is maximal when $s=K_M$ and approaches zero at very high or very low substrate concentrations:
\begin{equation*}
    \lim_{s\to +\infty}\eta(s)=0,\qquad \lim_{s\to 0+} \eta(s) =0.
\end{equation*}
Consequently, when $e_0\ll K_M$, ~\citet{Thomas2012} found that the stochastic QSSA introduced by~\citet{Rao2003} is only accurate at extremely high or low substrate concentrations unless $k_2\ll k_{-1}$.

Although the combined research of~\citet{Sanft} and~\citet{Thomas2011} established that the stochastic QSSA \eqref{stochMM} is not accurate under the same conditions that ensure the accuracy of the deterministic QSSA \eqref{QSSA}, neither work provided any indication as to {\it why} this occurs. Our analysis of linear reaction networks provides an explanation. To start, we begin by defining the slow variable ``$z$," in accordance with Fenichel theory discussed in~\cite{Noethen2011,Wechselberger2020}. The layer problem is formulated by setting $\varepsilon=0$ in \eqref{mmMASS}:
\begin{subequations}\label{mmL}
\begin{align}
\dot{s} &= \;\;\;k_1cs + k_{-1}c,\\
\dot{c} &= -k_1cs -k_{-1}c-k_2c.
\end{align}
\end{subequations}
The critical manifold is the $s$-axis: $S_0:=\{(s,c)\in \mathbb{R}^2:c=0\}$. Each point, $(s_b,0)\in S_0$ is called a base point, and for each base point there exists an invariant\footnote{Invariant in this context means invariant with respect to the layer problem.} fast fiber, $\mathcal{F}_0(\cdot,\cdot)$\footnote{The union of fast fiber defines the stable manifold of $S_0$: $W^s(S_0):=\bigcup_{(s_b,0)\in S_0}\mathcal{F}_0(s_b,0)$.}
\begin{equation}
\mathcal{F}_0(s_b,0):=\{(s,c)\in \mathbb{R}^2: s+c+K\cdot\ln\left(\cfrac{s+K_S}{s_b+K_S}\right) = s_b\}, \quad K_S :=k_{-1}/k_1,\quad K:=k_2/k_1
\end{equation}
The map $z=\zeta(s,c)=s+c+K\cdot\ln(s+K_S)$ is nonlinear, but  we can identify phase-plane regions where the fibers (and therefore $z$) are linear functions of $s$ and $c$. At very high and very low substrate concentrations, these fibers are asymptotically linear, and correspond to the regions in the phase-plane where the heuristically reduced CME \eqref{stochMM} is highly accurate. To show this, let us start with the latter case and first consider small $s$. In this region, it holds that
\begin{multline}
s+c+K\cdot\ln (s+K_S) = s+c+K\cdot \ln (K_S) + K\cdot \ln (1+s/K_S) \\= s_b + K\ln (K_S) + K\cdot \ln (1+s_b/K_S),
\end{multline}
which reduces to
\begin{equation*}
s+c + K\cdot \ln (1+s/K_S) = s_b +  K\cdot \ln (1+s_b/K_S).
\end{equation*}
If $s,s_b\ll K_S$, then $\ln (1+s_b/K_S)\sim s_b/K_S$ and $\ln (1+s/K_S)\sim s/K_S.$ This leaves us with
\begin{equation*}
s+c+ \cfrac{k_{-1}}{k_2}s = s_b + \cfrac{k_{-1}}{k_2}s_b,
\end{equation*}
and therefore at very low substrate concentrations we have
\begin{equation}\label{smallS}
\mathcal{F}(s_b,0) \sim \bigg\{(s,c)\in\mathbb{R}^2:s+\left(\cfrac{k_{-1}}{k_2+k_{-1}}\right)c = s_b\bigg\},
\end{equation}
which is linear in $s$ and $c$.

Next, consider large substrate concentrations. By similar arguments (namely, we assume $s,s_b \gg K_S$) we have
\begin{equation*}
s+c + K\cdot\ln(S+K_S) \sim s+c+K\cdot\ln (s) \sim s_b + K\cdot \ln(s_b).
\end{equation*}
It follows that because
\begin{equation*}
\lim_{s\to \infty}\cfrac{\ln(s^K)}{s}=0 \implies K\cdot \ln(s) \sim \smallO(s)
\end{equation*}
the logarithmic terms are negligible in comparison to the linear term, $s$. Thus, at very high substrate concentrations we obtain
\begin{equation}\label{Sfib}
\mathcal{F}(s_b,0) \sim \{(s,c)\in \mathbb{R}^2: s+c=s_b\}.
\end{equation}
which is also linear in $s$ and $c$. A consequence of \eqref{Sfib} is that, to leading order, there is no product formation on the fast timescale when $s$ is large.

Next, we illustrate {\it why} the heuristic CME \eqref{stochMM} prevails at very high and very low substrate concentrations. Beginning with the latter, the Michaelis-Menten mass action equations are well-approximated by their linearization about $(s,c)=(0,0)$ when the substrate concentration is extremely small compared to the Michaelis constant, $K_M$,~\cite{FirstOrder}:
\begin{subequations}
\begin{align}
\dot{s} &= -k_1e_0s +k_{-1}c,\\
\dot{c} &= \;\;\;k_1e_0s -k_{-1}c-k_2c.
\end{align}
\end{subequations}
When $e_0$ (or $k_1$) is taken as the singular perturbation parameter (we will consider $e_0\mapsto \varepsilon e_0$), we obtain the following {\it linear} singularly perturbed system
\begin{subequations}\label{linQSS}
\begin{align}
\dot{s} &= -\varepsilon {k}_1e_0s +k_{-1}c,\\
\dot{c} &= \;\;\;\varepsilon {k}_1e_0s -k_{-1}c-k_2c.
\end{align}
\end{subequations}
The form of \eqref{linQSS} is identical to \eqref{sys1} with
\begin{equation*}
A_0 = \begin{pmatrix}0&k_{-1}\\0 & -k_{-1}-k_2\end{pmatrix}, \quad v_+^{(0)}=\begin{pmatrix}1\\0\end{pmatrix}, \quad w_0=\begin{pmatrix}1\\ \gamma\end{pmatrix}, \quad \gamma:=\cfrac{k_{-1}}{k_{-1}+k_2}
\end{equation*}
Although the slow variable in this case is $z=s+\gamma c$, it is clear that, on the critical manifold $c=0$, we have ${\rm d}s/{\rm d}z=1$. Moreover, the first order approximation to $\lambda_+(\varepsilon)$ is
\begin{equation*}
\varepsilon \lambda_+^{(1)} = \varepsilon \cfrac{w_0^TA_1v_+^{(0)}}{w_0^Tv_+^{(0)}} =-\varepsilon \cfrac{k_2e_0}{K_M}.
\end{equation*}
Thus, the reduced LNA for substrate (at low concentration) is
\begin{subequations}
\begin{align}
s' &= -\cfrac{k_2e_0}{K_M}\cdot s = z'\\
{\rm d}X_s &= -\cfrac{k_2e_0}{K_M}\cdot X_s\;{\rm d}\tau -\sqrt{\cfrac{k_2e_0}{K_M}\cdot s}\;{\rm d}W(\tau) = {\rm d}Z,
\end{align}
\end{subequations}
and therefore, by Proposition \ref{prop3}, the reduced CME,
\begin{equation}\label{lstochMM}
\cfrac{\partial}{\partial t}Pr(n,t) = \cfrac{k_2e_0}{K_M}(E_n^{+1}-1)nPr(n,t),
\end{equation}
will accurately estimate the first two moments of $n$, which explains why the heuristic reduction \eqref{stochMM} retains precision when the substrate concentration is low. The reduced CME \eqref{lstochMM} is interpretable as a special case of \eqref{stochMM}.

\begin{remark}
From the perspective of singular perturbation parameters, it can be argued that the validity of the heuristically reduced CME \eqref{stochMM} holds in the limit as $k_1\to 0$ while keeping all other parameters fixed and constant; see~\citet{Goeke2015} for a precise definition of singular perturbations and Tikhonov-Fenichel parameter values (TFPVs). 
\end{remark}

Next we consider large $s$. The perturbed equations for small $e_0\mapsto \varepsilon e_0$ are
\begin{subequations}\label{SPA}
\begin{align}
\dot{s} &=-k_1(\varepsilon e_0-c)s + k_{-1}c,\\
\dot{c} &= \;\;\;k_1(\varepsilon e_0-c)s -k_{-1}c-k_2c.
\end{align}
\end{subequations}
and the critical manifold is {\it still}
\begin{equation*}
S_0 = \{(s,c)\in \mathbb{R}^2: c=0\}.
\end{equation*}
However, nothing about {\it large $s$} is overtly encoded in \eqref{SPA}. The projection matrix, $P_0$, is given by
\begin{equation*}
P_0 = \begin{pmatrix}1 & \cfrac{s+K_S}{s+K_M}\\0&0\end{pmatrix},
\end{equation*}
and therefore the Fenichel reduction is given by
\begin{equation}
\begin{pmatrix}s'\\c'\end{pmatrix} = \begin{pmatrix}1 & \cfrac{s+K_S}{s+K_M}\\0&0\end{pmatrix}\begin{pmatrix}-k_1e_0s\\\;\;\;k_1e_0s\end{pmatrix} = -\cfrac{k_2e_0s}{s+K_M}\begin{pmatrix}1\\0\end{pmatrix}
\end{equation}
which is of course the standard QSSA of the Michaelis-Menten reaction mechanism~\cite{Schnell1997}. But, how do we account for large $s$? The answer lies in the form of $P_0$, which can be expanded in terms of $\varepsilon_1 = K_S/s$ and $\varepsilon_2=K_M/s$:
\begin{equation}\label{Pexp}
P_0 = \begin{pmatrix}1 & \cfrac{s+K_S}{s+K_M}\\0&0\end{pmatrix}=
\begin{pmatrix}1&1\\0&0\end{pmatrix} + \begin{pmatrix}0&\varepsilon_1-\varepsilon_2\\0&0\end{pmatrix} + \mathcal{O}(\varepsilon_i\varepsilon_j), \quad\text{with}\;\; \varepsilon_1-\varepsilon_2 =-K/s.
\end{equation}
\begin{remark}
Note that the zeroth-order approximation to $P_0$ is
\begin{equation*}
P_0 = \begin{pmatrix}1\\0\end{pmatrix}\begin{pmatrix}1&1\end{pmatrix} + \mathcal{O}(\varepsilon_1,\varepsilon_2),
\end{equation*}
which is the outer product of the vector $(1\;\;0)^T$ that is tangent to $S_0$, with $(1\;\;1)$. The row vector $(1\;\;1)$ is the derivative of the map $z=\varphi(s,c)=s+c$,
\begin{equation*}
D\varphi(s,c)=(1\;\;1)
\end{equation*}
that coincides with the linear fiber \eqref{Sfib}. Thus, again we see that to leading order in $\varepsilon,\varepsilon_1,\varepsilon_2$, the fast fiber is linear in $s$ and $c$ and the slow variable is $z=s+c$.
\end{remark}
Next, if we then compute the Fenichel reduction up to leading order in $\varepsilon,\varepsilon_1$ and $\varepsilon_2$ we have
\begin{equation}
\begin{pmatrix}
s'\\c'
\end{pmatrix} = -k_2e_0\begin{pmatrix}1\\0
\end{pmatrix} + \mathcal{O}(\varepsilon_i\varepsilon_j).
\end{equation}

To complete the analysis, we now proceed to reduce the SDE that characterizes the fluctuations, $X_s,X_c$, about the mean field
\begin{multline}\label{ssL}
\begin{pmatrix}
{\rm d}X_s\\{\rm d}X_c
\end{pmatrix} = \begin{pmatrix}\;\;\;k_1c & \;\;\;k_1(s+K_S)\\-k_1c & -k_1(s+K_M)
\end{pmatrix}\begin{pmatrix}
X_s\\X_c \end{pmatrix}{\rm d}t+ \begin{pmatrix}-k_1e_0\\\;\;\;k_1e_0 
\end{pmatrix}X_s\;{\rm d}\tau \\+ \begin{pmatrix}-\sqrt{k_1(\varepsilon e_0-c)s}& \sqrt{k_{-1}c}&0\\\sqrt{k_1(\varepsilon e_0-c)s}&-\sqrt{k_{-1}c}&-\sqrt{k_2c}\end{pmatrix}{\rm d}{\bf W}(t)
\end{multline}
where ${\rm d}{\bf{W}}:=({\rm d}W_1,\;\;{\rm d}W_2\;\;{\rm d}W_3)^T$. Substitution of $c=\varepsilon e_0s/(s+K_M)$ into \eqref{ssL}, which is the first-order approximation to the slow manifold, $S_{\varepsilon}$, yields (to leading order in $\varepsilon_1$ and $\varepsilon_2$)\footnote{The details of this simplification can be found in the Appendix.}
\begin{equation}\label{redSS}
\begin{pmatrix}
{\rm d}X_s\\{\rm d}X_c
\end{pmatrix} = \begin{pmatrix}0 & \;\;\;k_1s\\0& -k_1s
\end{pmatrix}\begin{pmatrix}
X_s\\X_c \end{pmatrix}{\rm d}t +\sqrt{\varepsilon}\begin{pmatrix}0& \sqrt{k_{-1}e_0}&0\\0&-\sqrt{k_{-1}e_0}&-\sqrt{k_2e_0}\end{pmatrix}{\rm d}{\bf W}(t).
\end{equation}
Setting $\varepsilon =0$ reveals the critical manifold
\begin{equation*}
S_0 :=\{(X_s,X_c)\in \mathbb{R}^2: X_c=0\},
\end{equation*}
and the projection matrix in this case is (again, to leading order)
\begin{equation*}
P_0 = \begin{pmatrix}1&1\\0&0\end{pmatrix}.
\end{equation*}
Projecting \eqref{redSS} onto the center subspace of the drift matrix yields
\begin{subequations}
\begin{align}
s' &= -k_2e_0 = z',\\
{\rm d}X_s &= -\sqrt{k_2e_0}\;{\rm d}W(t) = {\rm d}Z.
\end{align}
\end{subequations}
Since $p=-z$, the reduced CME for product formation is
\begin{equation}\label{PCME}
\cfrac{\partial}{\partial t}Pr(n,t) =  v_{\infty}(E_n^{-1}-1)Pr(n,t), \quad v_{\infty}:=k_2e_0,
\end{equation}
which has a well-known solution\footnote{If the units here look suspicious, recall that we set $\Omega=1$, so to be consistent one would include the system size and write $v_{\infty}\mapsto \Omega v_{\infty}$ in \eqref{PCME} and \eqref{DIST}.}
\begin{equation}\label{DIST}
Pr(n,t) = \cfrac{(v_{\infty}t)^n}{n!}e^{\displaystyle -v_{\infty}t}.
\end{equation}

Both \eqref{PCME} and \eqref{lstochMM} can be seen as special cases of \eqref{stochMM}; see~\cite{JSEssLNA} for an additional discussion. It is important to note the conditions that ensure the accuracy and precision of the heuristic CME \eqref{stochMM} as $e_0\to 0$. First, when the mass action equations are restricted to lie near $S_{\varepsilon}$, all diffusion occurs on the slow timescale, $\tau$ (again, see \citet{JEWS2026} for additional details). Second, the fibers that foliate $W^s(S_0)$ are, for all intents and purposes, linear in the regions where the heuristic CME is valid. Third, ${\rm d}s/{\rm d}z=1$ holds along the critical manifold in regions where the concentration of $s$ is very large or very small. This is not true in the intermediate region where $s\propto K_M$. In fact, constructing a diffeomorphism (a smooth, invertible coordinate transformation) from the first integral~\cite{Noethen2011}
\begin{equation*}
s+c+K\cdot \ln (s+K_S) = const. :=z
\end{equation*}
reveals that the temporal behavior of the slow variable ``$z$" depends nonlinearly on $s$
\begin{equation*}
z' = -\cfrac{v_{\infty}s}{s+K_S}.
\end{equation*}
On the slow timescale we have
\begin{equation}
s' = z' \cdot \cfrac{{\rm d}s}{{\rm d}z} = z'\cdot \cfrac{s+K_S}{s+K_M } = -\cfrac{v_{\infty}s}{s+K_M},
\end{equation}
which is the standard QSSA. Hence, the mean field (average) dynamics for $z$ and $s$ differ since ${\rm d}s\neq {\rm d}z$ in phase plane regions where $s\propto K_M$. 

At intermediate substrate concentrations on the order of $K_M$, additional requirements are needed to ensure the accuracy and precision of the heuristic CME \eqref{stochMM}. Specifically, it is necessary to ensure $\eta(s) \ll 1$, as $e_0\to 0$ alone does not induce small $\eta(s)$. If we think of $\eta(\cdot)$ not only as a function of $s$ but also as a function of the rate constants, $k_i$, then we can state the following:
\begin{proposition}
Let $e_0\mapsto \varepsilon e_0$ and $k_i\mapsto \varepsilon ^{\pm 1}k_i$. If 
\begin{equation}
\lim_{\varepsilon \to 0} \eta(s,\varepsilon) =0
\end{equation}
where $k_i$ is a rate constant, then the fibers $\mathcal{F}_0(\cdot)$ that foliate $W^s(S_0)$, are linear in $s$ and $c$. 
\end{proposition}
\begin{proof}
The term $\eta(s)$ vanishes in the following limits (see {\sc table} $1$):
\begin{table}[htbp]
\centering
\label{tab:my_table}
\begin{tabular}{|c||c||c|}
\hline
\textbf{Singular perturbation parameters} & \textbf{$\lim_{\varepsilon \to0} \eta(s,\varepsilon)$} & \textbf{$\mathcal{F}_0(\cdot)$} \\ \hline \hline
$e_0,k_1\mapsto \varepsilon e_0,\varepsilon k_1$     & 0    & $\{(s,c)\in \mathbb{R}^2: s+\frac{k_{-1}}{k_{-1}+k_2}c = s_b\}$     \\ \hline
$e_0,k_1\mapsto \varepsilon e_0,\varepsilon^{-1} k_1$    & 0    & $\{(s,c)\in \mathbb{R}^2: s+c = s_b\}$    \\ \hline
$e_0,k_2\mapsto \varepsilon e_0,\varepsilon k_2$   & 0    & $\{(s,c)\in \mathbb{R}^2: s+c = s_b\}$    \\ \hline
$e_0,k_2\mapsto \varepsilon e_0,\varepsilon^{-1} k_2$     & 0    & $\{(s,c)\in \mathbb{R}^2: s = s_b\}$    \\ \hline
$e_0,k_{-1}\mapsto \varepsilon e_0,\varepsilon^{-1} k_{-1}$     & 0    & $\{(s,c)\in \mathbb{R}^2: s+c = s_b\}$  \\ \hline
\end{tabular}
\caption{When the fibers that foliate $W^s(S_0)$ are linear, the term $\eta(s)$ becomes asymptotically negligible and the heuristically reduced CME \eqref{stochMM} is valid.}
\end{table}\label{Table1}
\end{proof}

Thus, in a colloquial sense, the term $\eta(s,k_i)$ provides a notion of the linearity of the fast fibers. Again, the linearity of the fibers is critical to the success of heuristically reduced CMEs. This is very apparent in the chemical Langevin equation (CLE) regime, where the SDE of the reaction network is driven by multiplicative noise and where nonlinear fibers can introduce additional drift terms~\cite{Katz,Parsons2017}.

The major takeaway from our analysis of the Michaelis-Menten network lies within the following observation: In regions where the stochastic QSSA \eqref{stochMM} is valid, the fast fibers (to leading order in $\varepsilon,\varepsilon_i$, etc..) align with the linearized fiber defined by the image of the Jacobian along $S_0$. Consider the region in which $e_0$ is small and $s\ll K_S,K$. The Jacobian along the critical manifold is, to leading order in $e_0$ and $s/K,s/K_S$,
\begin{equation*}
\begin{pmatrix}0 & \;\;k_{-1}\\0&-k_{-1}-k_2\end{pmatrix},\quad w_0 :=\begin{pmatrix}1\\\gamma\end{pmatrix}
\end{equation*}
and thus the fast fibers \eqref{smallS} align with the image of the Jacobian along $S_0$; see {\sc figure} \ref{FIG4}.
\begin{figure}[bth!]
    \centering
    \includegraphics[scale=0.85]{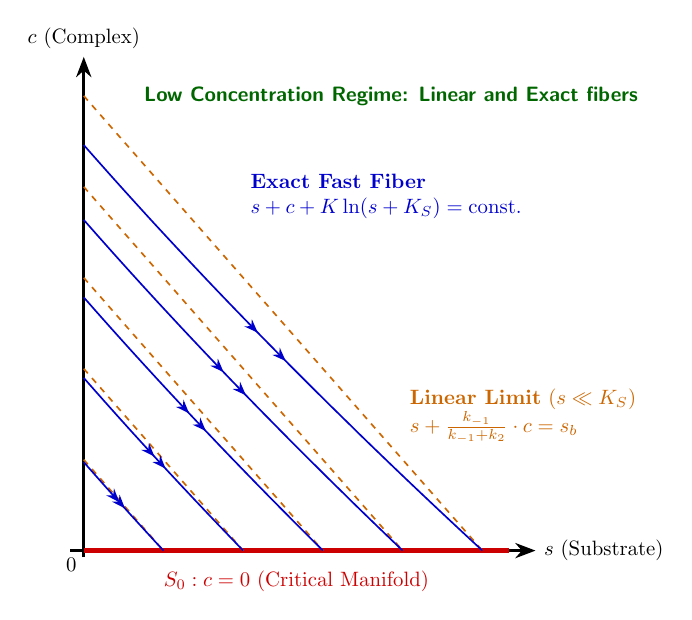}
    \caption{\textbf{When $s$ is very low in concentration, the fast fibers that foliate the critical manifold begin to align with the linearized fibers computed from the Jacobian along the critical manifold.} Towards the right of the phase-plane, the concentration of $s$ is proportional to $K_M$, and the linearized fibers only approximate the exact fiber near the critical manifold. Towards the origin, the linearized fibers start to align with the exact fiber as the concentration of $s$ dips far below $K_M$.}
    \label{FIG4}
\end{figure}

The same holds for large values of $s$ (i.e., when $s\gg K_M$). The Jacobian along $S_0$ is (again, to leading order in $e_0,\varepsilon_1,\varepsilon_2$) 
\begin{equation*}
\begin{pmatrix}0 & \;\;\;k_1s\\0&-k_1s\end{pmatrix},\quad w_0 :=\begin{pmatrix}1\\1\end{pmatrix}
\end{equation*}
and thus, the fast fibers \eqref{Sfib} align with the image of the Jacobian evaluated along the critical manifold; see {\sc figure} \ref{FIG5}.
\begin{figure}[bth!]
    \centering
    \includegraphics[scale=.85]{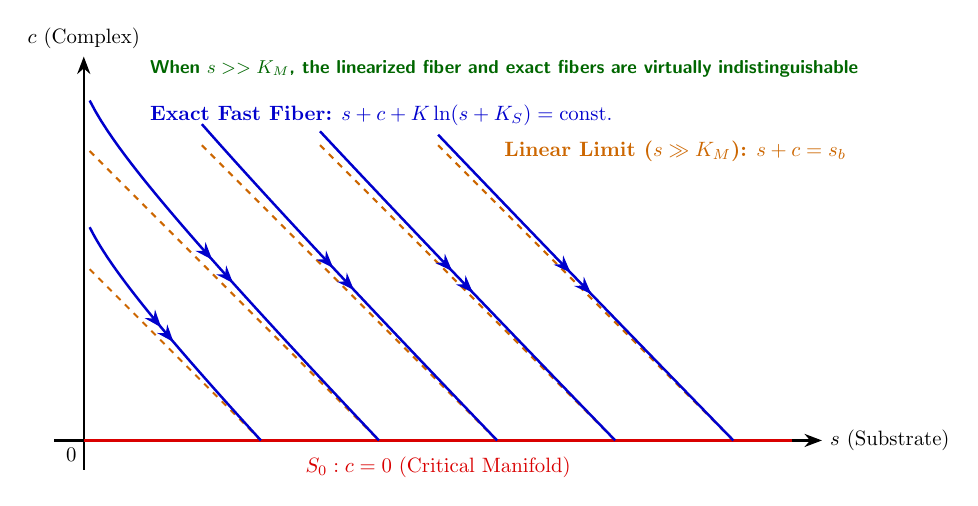}
    \caption{\textbf{When $s$ is very high in concentration the fast fibers that foliate the critical manifold begin to align with the linearized fibers computed from the Jacobian along the critical manifold.} At the left of the phase-plane, concentrations of $s$ are proportional to $K_M$, and we see that the linearized fibers only approximate the exact near the critical manifold. As the concentration of $s$ exceeds $K_M$, the exact fibers begin to align with the linear limit: $s+c=s_0$. }
    \label{FIG5}
\end{figure}

\subsection{Linearization, system size expansion, and bimolecular reactions}

Consider the reaction network \eqref{GC} with small $k_1$. Linearization about $(a,b)=(0,0)$ yields
\begin{subequations}\label{GCMF}
\begin{align}
\dot{a} &= -\varepsilon k_1a, \\
\dot{b} &=\;2\varepsilon k_1a -k_2 b,
\end{align}
\end{subequations}
while fluctuations in the mean field \eqref{GCMF} satisfy
\begin{equation}
\begin{pmatrix}{\rm d} X_{a}\\{\rm d} X_{b}\end{pmatrix} = \begin{pmatrix}-\varepsilon k_1 & 0\\ 2\varepsilon k_1 &  -k_2\end{pmatrix}\begin{pmatrix}X_{a}\\X_{b}\end{pmatrix}{\rm d}t + \begin{pmatrix}-\sqrt{\varepsilon k_1 a} & 0\\ 2\sqrt{\varepsilon k_1a}& -\sqrt{k_2b}\end{pmatrix}\begin{pmatrix}{\rm d}W_1(t)\\{\rm d}W_2(t)\end{pmatrix}.
\end{equation}
Linearization automatically brings the LNA into the standard form. Moreover, the dynamics of $a$ and $X_{ a}$ do not depend on $b$ or $X_{b}$. The reduced equations for $a$ and $X_{a}$ are
\begin{subequations}\label{formR}
\begin{align}
a' &= -k_1a,\\
{\rm d}X_{a} &= -k_1X_{ a}{\rm d}\tau - \sqrt{k_1a}\;{\rm d}W(\tau),
\end{align}
\end{subequations}
which accurately approximates the first and second moments of the full LNA when $k_1$ is small and $(a,b)$ is close to the origin; see {\sc figure} \eqref{FIG6}, left panel.

The form of the reduced system \eqref{formR} suggests that the CME
\begin{equation}\label{GCstoch}
\cfrac{\partial}{\partial t}Pr(n,t) =  \varepsilon k_1(E_n^{+1}-1)nPr(n,t)
\end{equation}
which corresponds to the reduced reaction mechanism,
\begin{align}\label{GC0}
 \ce{A ->[$k_1$] 2C}
\end{align}
might be accurate and precise when $n$, the number of $A$ molecules, is comparatively small. The problem is that the linearization neglects terms that are $\mathcal{O}(b^2)$ near the origin. This is permissible in scenarios where the system size is extremely  {\it large}, since the nonlinear terms do not alter the qualitative dynamics near the hyperbolic origin. 

On the other hand, the second order reactions of nonlinear networks (such as homodimerization reactions) may not be negligible when the system size is small and the state space is discrete. For example, let $m$ denote the number of $B$ molecules. The propensity of the elementary reaction, ``$R_2$"
\begin{align}\label{GC2}
  R_2 := \ce{2B ->[$k_{-1}$] A}
\end{align}
is $\frac{1}{2}k_{-1}m(m-1)$, which is {\it not} negligible when the size of the system is small. Since \eqref{GC2} cannot be ignored, the reduced CME implicitly suggests that every time $R_1$ occurs, the reactions $R_2$ and $R_3$ in succession occur with probability $1$, where the elementary reactions $R_1$ and $R_3$ are 
\begin{subequations}
\begin{align*}
R_1&:= \ce{A ->[$\varepsilon k_{1}$] 2B},\\
R_3&:= \ce{B ->[$k_{2}$] C}.
\end{align*}
\end{subequations}
The question then becomes, under what conditions, if any, does \eqref{GC0} apply? Suppose a reaction occurs when $t=t_0$. The probability that $R_1$ occurs is
\begin{subequations}
\begin{align}
Pr(R_1)&= \begin{cases}
\mathcal{O}(\varepsilon), \quad n,m \geq 1,\\
\\
0,\qquad\;\; n=0.
       \end{cases}\\
Pr(R_1)&= \begin{cases}
1, \quad \;\;\quad n \geq 1,\;\;m=0\\
\\
0,\qquad\;\; n=0.
       \end{cases}
\end{align}
\end{subequations}
The probability that $R_3$ occurs is
\begin{equation}\label{prob1}
    Pr(R_3)=\begin{cases}\cfrac{k_2m}{\frac{1}{2}k_{-1}m(m-1)+k_2m + \varepsilon k_1n} = \cfrac{k_2}{\frac{1}{2}k_{-1}(m-1)+k_2} - \mathcal{O}(\varepsilon),\quad m \geq 2,\;n\geq 0,\\
    \\
    1 - \mathcal{O}(\varepsilon),\qquad m=1,\;\;n\geq 0,\\
    \\
    0,\qquad \qquad \;\quad m=0,\;\;n\geq 0.
    \end{cases}
\end{equation}
Note that if $k_{-1}\sim \mathcal{O}(\varepsilon)$ then we have 
\begin{equation}\label{prob2}
    Pr(R_3)=\begin{cases} 1 - \mathcal{O}(\varepsilon) \;\;\qquad m \geq 2,\;\;n\geq 0,\\
    \\
    1 - \mathcal{O}(\varepsilon),\qquad m=1,\;\;n \geq 0,\\
    \\
    0,\qquad \qquad \;\quad m=0,\;\;n\geq 0.
    \end{cases}
\end{equation} 
Moreover, defining $k_{-1}\mapsto \varepsilon k_{-1}$ leaves us with
\begin{equation}
Pr(R_2)= \begin{cases}
\mathcal{O}(\varepsilon), \quad m \geq 2,\\
\\
0,\qquad\;\; m\leq 1.
       \end{cases}
\end{equation}

Although its derivation is based on heuristics, the distribution \eqref{prob2} illustrates the following: if $k_1\sim \mathcal{O}(\varepsilon)$ {\it and} $k_{-1} \sim \mathcal{O}(\varepsilon)$, then nearly every reaction that occurs will be $R_3$ as long as $m\neq 0$. Moreover, based on this heuristic analysis, one can speculate that the reduced CME \eqref{GCstoch} is accurate as long as $k_1$ and $k_{-1}$ are $\mathcal{O}(\varepsilon)$, a claim that is easily verified numerically; see {\sc figure} \ref{FIG6}.

If $k_{-1}$ is not $\mathcal{O}(\varepsilon)$, then it is still possible to construct a reduced CME, but it will differ from \eqref{GCstoch}. Thus, while the heuristically reduced CME \eqref{GCstoch} is only valid when $k_1$ and $k_{-1}$ are $\mathcal{O}(\varepsilon)$, it is possible to construct a highly accurate reduced CME, but one must analyze the CME directly; see~\citet{Mastny2007} for details. 

\begin{figure}[bth!]
    \centering
    \includegraphics[scale=0.50]{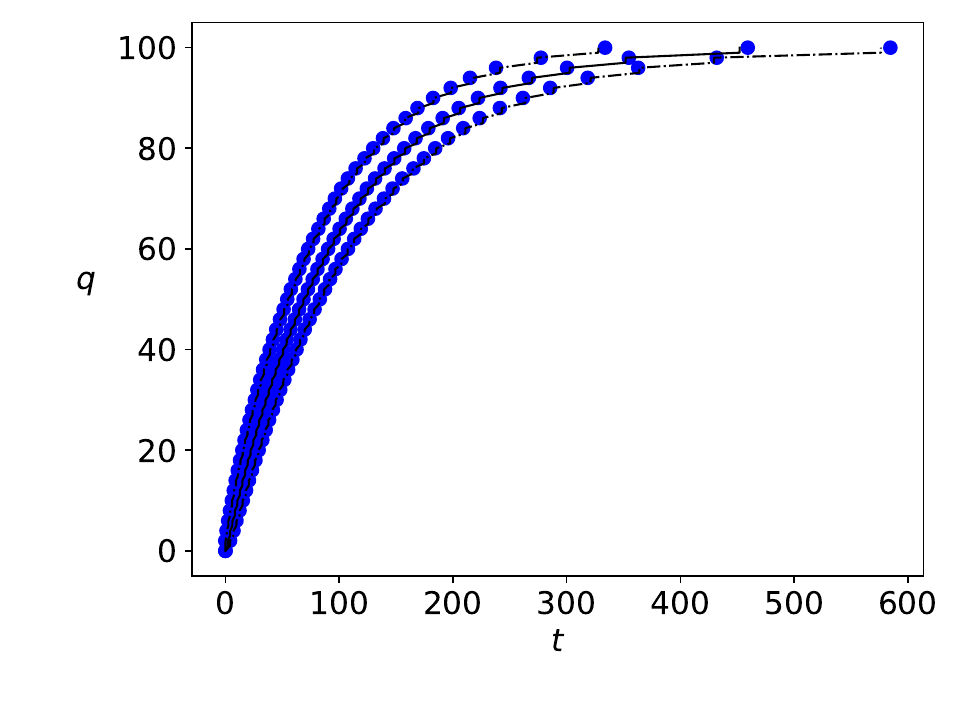}
    \includegraphics[scale=0.50]{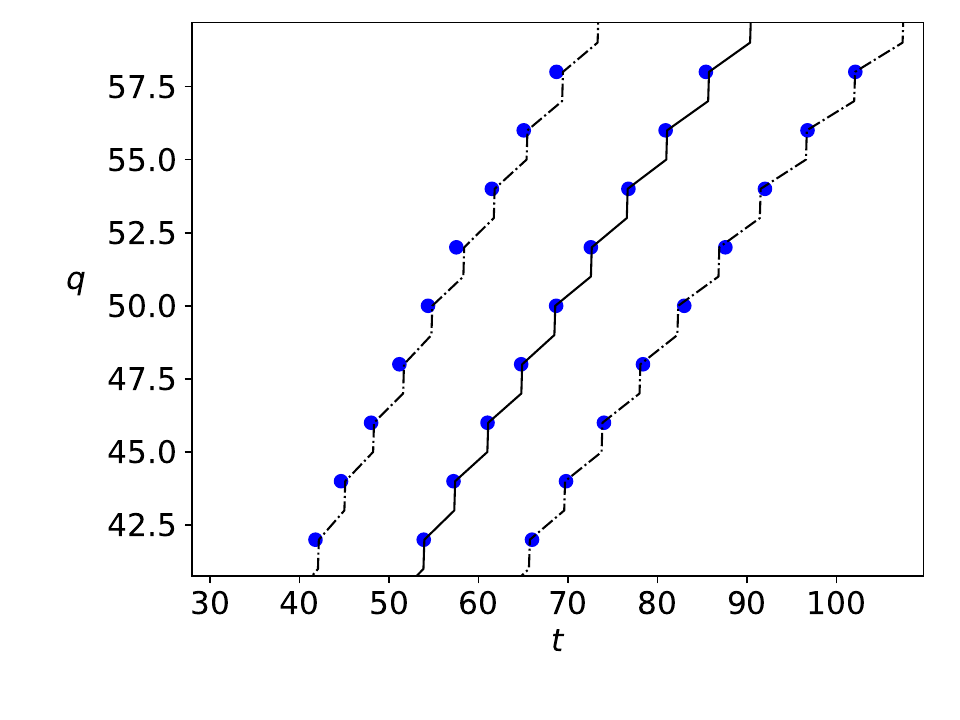}
    \caption{\textbf{The heuristically reduced CME \eqref{GCstoch} provides an excellent approximation to the first two moments of $q$, the number of $C$ molecules. } In this figure, the solid black line is the mean obtained from $10000$ simulations of the Gillespie algorithm applied to the full CME; the dotted black lines are $\pm$ one standard deviation. The blue circles demarcate the formation of $C$ obtained from the reduced CME \eqref{GCstoch}. Note that the full CME generates a single $C$ molecule each time the third elementary reaction occurs, but the reduced CME \eqref{GCstoch} generates {\it two} product molecules each time a reaction occurs. {\sc left panel}: In this simulation, $n(0)=50,m(0)=0$ and $q(0)=0$ with (in arbitrary units) $k_1=k_{-1}=0.01$ and $k_3=10.0$. {\sc right panel:} A closeup of the left panel.}
    \label{FIG6}
\end{figure}

If we revisit the complete, nonlinear mass action system with  $k_1\mapsto \varepsilon k_1$ and $k_{-1}\mapsto \varepsilon k_{-1}$,
\begin{subequations}\label{Hdim}
\begin{align}
\dot{a} &= -\varepsilon k_1a + \varepsilon k_{-1}b^2,\\
\dot{b} &= \;2\varepsilon k_1a -2\varepsilon k_{-1}b^2 -k_2b,
\end{align}
\end{subequations}
we see that $a$ is indeed the slow variable ($a\equiv z$), the critical manifold is the $a$-axis, and that the fibers are linear: $\mathcal{F}_0(a_b,0):=\{(a,b)\in \mathbb{R}^2: a=a_b,b=b\}$. Thus, one expects that the heuristically reduced CME will prevail in this scenario, which it does (again, see {\sc figure} \eqref{FIG6}).

In contrast, with only small $k_1$, the singularly-perturbed mass action system is
\begin{subequations}\label{ABfull}
\begin{align}
\dot{a} &= -\varepsilon k_1a + k_{-1}b^2,\\
\dot{b} &= \;2\varepsilon k_1a -2k_{-1}b^2 -k_2b,
\end{align}
\end{subequations}
and the critical manifold is still the $a$-axis: $S_0:=\{(a,b)\in \mathbb{R}^2: b=0\}.$ However, the fibers that foliate $W^s(S_0)$ are nonlinear in this case:
\begin{equation}\label{ABFib}
\mathcal{F}_0(a_b,0) = \bigg\{(a,b)\in \mathbb{R}^2: \zeta(a,b)=a+\frac{1}{2}b-\frac{1}{4K_r}\ln \left(2K_rb+1\right)=a_b\bigg\}, \quad W^s(S_0) = \bigcup_{(a_b,0)\in S_0}\mathcal{F}_0(a_b,0),
\end{equation}
where $K_r:=k_{-1}/k_{2}$. Although these fibers are asymptotically linear near the origin, the homodimerization reaction ($R_2$) in the low volume limit cannot be ignored.\footnote{By comparison, recall that the bimolecular binding reaction is not overtly eliminated in the linearized Michaelis-Menten network, it is simply replaced by the first-order reaction: $k_1(e_0-c)s\to k_1e_0s$.}

The real takeaway from this example is the following: In the deterministic regime, the Fenichel reduction of the fully nonlinear mass action equations for $a$ is the same, regardless of whether $k_{-1}$ is $\mathcal{O}(\varepsilon)$ or $\mathcal{O}(1)$:
\begin{equation*}
a' = -k_1a.
\end{equation*}
Moreover, if we drop one rung on the thermodynamic ladder and compute the Fenichel reduction for $X_a$, we find once again that the reduction is the same, regardless of how $k_{-1}$ scales:
\begin{equation*}
{\rm d}X_a = -k_1X_a\;{\rm d}\tau-\sqrt{k_1a}\;{\rm d}W(\tau).
\end{equation*}
Thus, we see {\it no} difference in the {\it leading order} terms in the reduction of $a$ at the deterministic or linear noise regimes when $k_1\mapsto \varepsilon k_1$, irrespective of whether or not $k_{-1}\sim \mathcal{O}(\varepsilon)$ or $k_{-1}\sim\mathcal{O}(1)$. Thus, as long as $k_1$ is sufficiently small, the first two moments of $a$ are well-approximated by \eqref{moments}; see {\sc figure} \ref{FIG7}.
\begin{subequations}\label{moments}
\begin{align}
a' &= -k_1a,\label{M1}\\
{\rm Var}(X_a)' &= -2k_1 {\rm Var}(X_a) + k_1a.\label{M2}
\end{align}
\end{subequations}

\begin{figure}[htb!]
    \centering
    \includegraphics[scale=0.50]{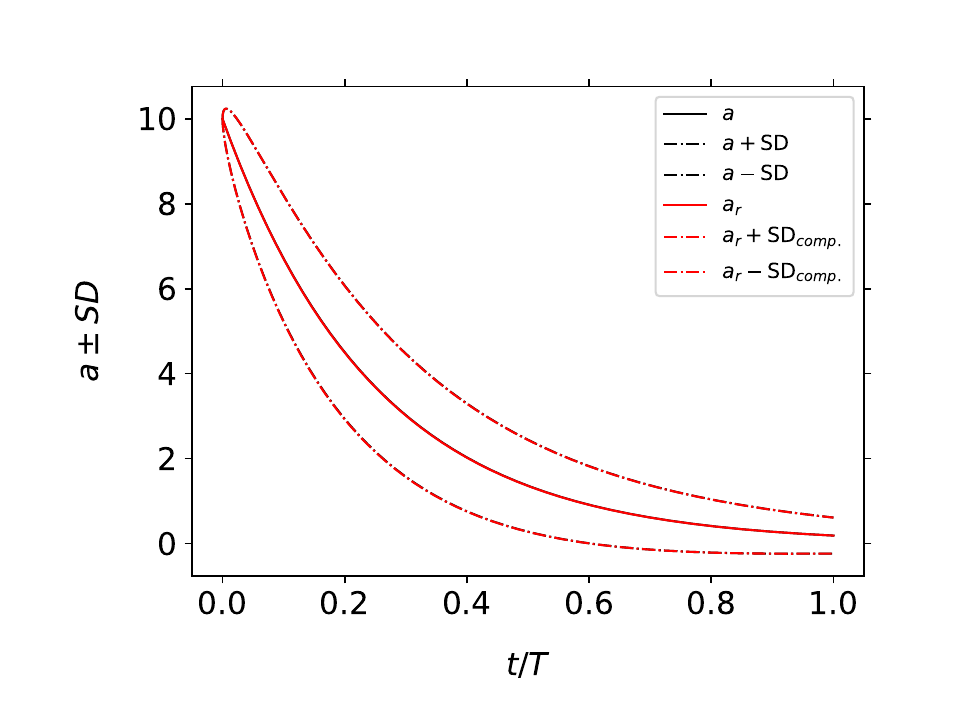}
    \includegraphics[scale=0.50]{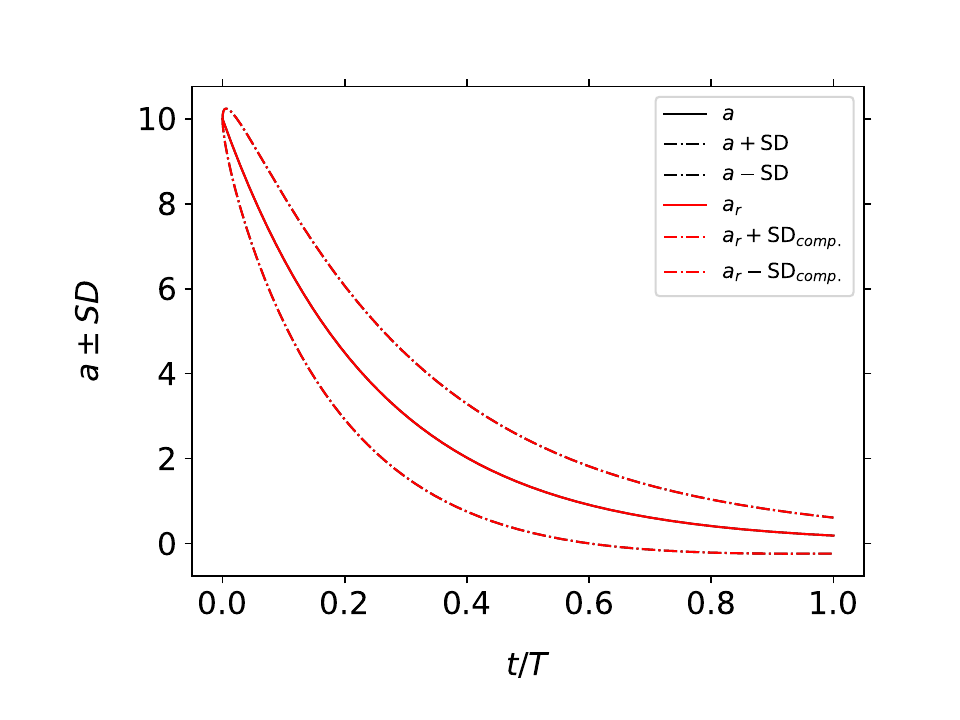}
    \caption{\textbf{The reduced moment equations \eqref{moments} accurately approximate the mean $\mathbb{E}(a)$ and variance of ${\rm Var}(X_a)$ as $k_1 \to 0$.} In this figure, the solid black line is the mean of $a$ obtained by numerically solving the mass action equations; the dashed/dotted black lines are $\pm$ one standard deviation obtained from the numerical solution to full covariance equation. The solid red line is the numerical solution to the reduced equation for $\mathbb{E}(a)$ given by \eqref{M1} (this solution is by denoted $a_r$ in the legend). The dashed/dotted red curves are $\pm$ one standard deviation obtained from \eqref{M2}. For aesthetic purposes time is scaled by $t/T$, where $T$ is time interval of the simulation. {\sc left panel}: In this simulation, $a(0)=10,b(0)=0$ and (in arbitrary units) $k_2=k_{-1}=10.0$ and $k_1=0.001$. {\sc right panel:} In this simulation, $a(0)=10,b(0)=0$ and (in arbitrary units) $k_2=10.0$ and $k_{-1}=k_1=0.001$. Observe that not only are the reduced equations highly accurate, but moreover the magnitude of $k_{-1}$ has no real impact on the dynamics as long as $k_{1}$ is sufficiently small.}
    \label{FIG7}
\end{figure}
Thus, for a large enough system size, the magnitude of $k_{-1}$ has little impact on the long-time dynamics as $k_1 \to 0$. This happens for several reasons, but three primary reasons are as follows. First, the critical manifold $S_0$ is the same in both cases and corresponds to $b=0$ (the $a$-axis) and therefore the magnitude of $b$ in the vicinity of the slow manifold is $\mathcal{O}(\varepsilon)$. Thus, quadratic terms such as $b^2$ are $\mathcal{O}(\varepsilon^2)$ near $S_0$ and therefore are asymptotically negligible. Second, the linearized fibers of the layer problem associated with \eqref{ABfull} coincide with the fibers of the layer problem associated with \eqref{Hdim}; see {\sc figure} \ref{FIG8}. 
\begin{figure}[bth!]
    \centering
    \includegraphics[scale=0.85]{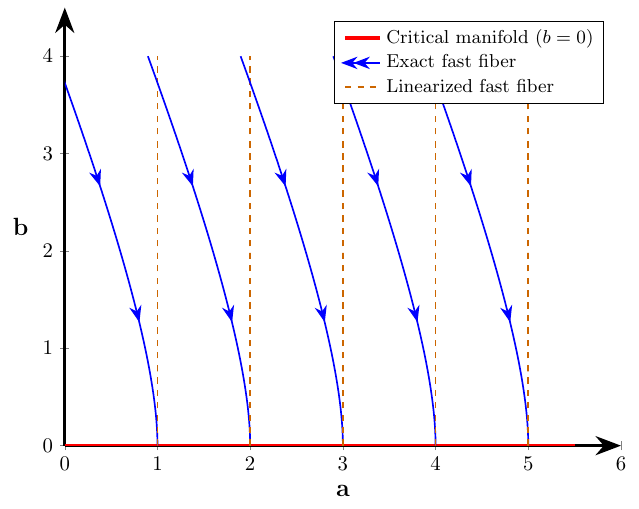}
    \caption{\textbf{The linearized fast fibers are given by $a=a_b$ when $(k_1,k_{-1})\mapsto (\varepsilon k_1,k_{-1})$ but, unlike the Michaelis-Menten system at high and low substrate concentrations, the exact fast fiber does not coincide with the linearized fibers}. In this figure, the critical manifold is the $a$-axis, with several base points demarcated by black circles. The fast fibers (from \eqref{ABFib}) associated with each base point are the blue curves. The orange vertical lines are the tangent lines to the fibers (linearized fast fibers) at each base point.}
    \label{FIG8}
\end{figure}
Third, the zeroth- and first-order terms in the asymptotic expansion of the slow manifold, $S_{\varepsilon}$, are identical
\begin{equation}
{\rm graph}(S_{\varepsilon}) =\{(a,b)\in \mathbb{R}^2_{\geq 0}: b= 2\varepsilon k_1a/k_2+\mathcal{O}(\varepsilon^2)\}
\end{equation}
and in this case diffusion is limited to the slow timescale; again, see Proposition \ref{prop5}.

This example highlights the fact that there can be no difference (at least up to leading order) between the Fenichel reductions for different perturbations of the mass action system and their associated LNA. However, as the system size is reduced, the changes become significant. Homodimerization reactions and other second-order reactions are examples of elementary reactions that can occasionally be ignored (in the asymptotic sense) when the system size is adequately large but, as the system size shrinks, these same reactions may significantly influence the stochastic timecourse dynamics.

\subsection{Discussion}

Our intent with this work was not only to determine the conditions that determine when heuristic reductions of CME retain accuracy and precision but also to underscore the fact that the breakdown of heuristic reductions can often be traced back to geometric factors present in the deterministic analysis. The heuristic reduction of the CME for nonlinear reaction networks is a difficult endeavor, and it is unclear whether or not there exists a consistent theory that allows for an a priori determination of when reduction -- in any form -- is possible. The linear analysis of stochastic reaction networks provides geometric insight as to why certain heuristically reduced CMEs often fail to approximate the moments of the unreduced (complete) CME. What our linear analysis strongly suggests, and what is implied by similar studies on the topic~\cite{Parsons2017}, is that the fast fibers that foliate the stable manifold of the critical manifold {\it must}, at the very least, be linear. 
Furthermore, by Proposition \ref{prop5}, the variance of fluctuations must be $\mathcal{O}(\varepsilon)$ on the fast time scale when the mass action equations are restricted to evolve on (or sufficiently close to) the slow manifold. Mathematically, we have the following commonalities:
\begin{enumerate}
\item The parametric variable of interest, $x_i$, should evolve with the same velocity as $z$ (the slow variable) on the slow timescale, $\tau$: ${\rm d}x/{\rm d}z=1$.\\
\item The slow variable must be representable as an integer combination of molecular counts at low volume: the quantity $s+c+ K\cdot \ln (s+K_S)$ is not. 
\end{enumerate}

For the Michaelis-Menten reaction network, small $e_0$ ensures that the critical manifold is the $s$-coordinate axis and therefore diffusion is limited to the slow timescale (see Proposition \ref{prop5} and~\cite{JEWS2026} for additional details), but the fibers that foliate $W^s(S_0)$ are nonlinear. Some of the confusion surrounding the nature of the fast fibers that foliate $W^s(S_0)$ of the singularly perturbed Michaelis-Menten network originates from the use of degenerate scaling, $\widehat{c}=c/\varepsilon$, which is heavily employed in deterministic analyses~\cite{Heineken1967,Segel1988,Segel1989,Borghans1996}. With degenerate scaling, one obtains the rescaled mass action system
\begin{subequations}\label{mscale}
\begin{align}
\dot{s} &= -\varepsilon k_1(e_0-\widehat{c})s + \varepsilon k_{-1}\widehat{c},\label{SD}\\
\dot{\widehat{c}} &= k_1(e_0-\widehat{c})s -k_{-1} \widehat{c}-k_2\widehat{c},\label{CD}
\end{align}
\end{subequations}
which suggests that the critical manifold, $S_0$, is the $\widehat{c}$-nullcline and that the fibers that comprise the stable manifold are {\it linear}: $\mathcal{F}(p)=\{(s,\widehat{c})\in \mathbb{R}^2: s=s_b\}$.

Although degenerate scaling suggests that $s$ is a slow variable, several problems emerge. First, as~\citet{Lax2020} point out, the map $(s,c,\varepsilon)\mapsto (s,\widehat{c},\varepsilon)$ is only a diffeomorphism if $\varepsilon >0$, and therefore -- mathematically speaking -- the substrate concentration does not qualify as the first integral of the layer problem. Second, as~\citet{OthmerI} mention, there are no {\it fast} and {\it slow} species, only fast and slow reactions. However, the degenerate scaling obfuscates this. For example, one cannot determine which elementary reactions are slow and fast from \eqref{mscale}, since the dissociative reaction is slow ($\mathcal{O}(\varepsilon)$) in \eqref{SD}, but fast ($\mathcal{O}(1)$) in \eqref{CD}. Third, the degenerate scaling conceals the fact that the actual fibers are curved in regions where $s\propto K_M$.

Finally, we point out that our analysis focused primarily on uncovering conditions that ensure the accuracy and precision of the heuristic CME when parameterization is with respect to the concentration of a {\it single} chemical species. Our motivation behind this choice is based entirely on practice: in most applications, one is interested in approximating the stochastic timecourse of a single species (such as substrate or product for enzymatic reactions). However, it is often possible to generate an accurate heuristic reduction of the CME\footnote{This concept was discussed in~\citet{Thomas2012}, as well as in~\cite{KangKim2017,kang2013}.} when the slow variable is a linear combination of species' concentrations in the sense that {\it stoichiometry is chemically interpretable} when passed to low volume environments
\begin{equation*}
z=w_0^Tx, \quad x\in \mathbb{R}^n, \quad w_0 \in\mathbb{Z}^n.
\end{equation*}
Linear integer combinations of species (such as the total substrate concentration~\cite{KIM2014,KIM2015,Kim2020,ganguly2025asymptotic,Burrage2008,Herath} or linear first integrals of the layer problem) have been shown to yield successful reductions, but a rigorous investigation of this concept in a more general context is left for future work.

\newpage
%%%%%%%%%%%%%%%%%%%%%%%%%%%%
\section*{Appendix}
This appendix consists of three subsections. In the first subsection we provide some of the basic definitions from geometric singular perturbation theory; in particular we define the {\it slow variable} in a more general sense as well as the {\it standard form} of a singularly perturbed differential equation.

In the second subsection, we do a complete multiscale analysis of the linear reaction \eqref{reac1} for small $k_3$, and address why the heuristically reduced CME for $n$, the molecular copy number of species $\bar{X}$, fails to approximate the variance of $n$. 

In the third subsection, we provide the details behind the reduction of the LNA \eqref{redSS}. Specifically, we illustrate (asymptotically) why high substrate concentration leads to a reduced LNA that describes a pure diffusion process.

\subsection*{Singular perturbations, slow variables, and the standard form}

Let $x\in \mathbb{R}^n$ and $y\in \mathbb{R}^m$. An ordinary differential equation of the form
\begin{subequations}\label{DN}
\begin{align}
\dot{x} &= f(x,y,\varepsilon),\quad f: \mathbb{R}^n\times \mathbb{R}^m \times [0,\varepsilon_0)\to \mathbb{R}^n,\\
\dot{y} &=g(x,y,\varepsilon),\quad  g: \mathbb{R}^n\times \mathbb{R}^m \times [0,\varepsilon_0)\to \mathbb{R}^m,
\end{align}
\end{subequations}
is singularly perturbed if the level set
\begin{equation*}
S_0 = \{(x,y)\in \mathbb{R}^n\times \mathbb{R}^m: F(x,y,0)=0\}, \quad F(x,y,\varepsilon):=\begin{pmatrix}f(x,y,\varepsilon)\\g(x,y,\varepsilon)\end{pmatrix}
\end{equation*}
is a nonempty set comprised of an infinite number of equilibrium points; for our purposes, we will only consider cases where $S_0$ is a $k$-dimensional submanifold of $\mathbb{R}^{n+m}$ with $0<k<n+m$. The set $S_0$, which is called the {\it critical manifold}, is normally hyperbolic if the eigenspectrum Jacobian, $DF(x,y)$, along $S_0$, is comprised of a zero (trivial) eigenvalue with geometric and algebraic multiplicity $k$, and $n+m-k$ eigenvalues that are strictly bounded away from the imaginary axis. If the $n+m-k$ nontrivial eigenvalues have negative real parts, then $S_0$ is attracting. Normal hyperbolicity implies that $S_0$ is structurally stable, and therefore turning on the perturbation results in the existence of a normally hyperbolic, invariant, and attracting slow manifold, $S_{\varepsilon}$.

The long-time dynamics of \eqref{DN} can be approximated by approximating the flow {\it on} $S_{\varepsilon}$. To do this, we expand $F(x,y,\varepsilon)$ in a Taylor series around $\varepsilon=0$
\begin{equation}
F(x,y,\varepsilon) = F(x,y,0) + \varepsilon DF(x,y,0) + \mathcal{O}(\varepsilon^2),
\end{equation}
and then project the vector field onto the tangent bundle of $S_0$,
\begin{equation}
\begin{pmatrix}\dot{x}\\\dot{y}\end{pmatrix} = \varepsilon P_0DF(x,y,0)|_{(x,y)\in S_0}+\mathcal{O}(\varepsilon^2), \quad P_0:\mathbb{R}^{n+m}\to \ker DF(x,y,0) \cong T_{(x,y)}S_0.
\end{equation}
It is common to define the slow timescale, $\tau =\varepsilon t$, in which case we have
\begin{equation}\label{SF}
\begin{pmatrix}x'\\y'\end{pmatrix} = P_0DF(x,y,0)|_{(x,y)\in S_0}
\end{equation}
where ``$\phantom{x}'$" denotes differentiation with respect to $\tau$. 

The flow on $S_{\varepsilon}$ converges to the flow associated with \eqref{SF} as $\varepsilon \to 0$, but the emphasis should be on the word {\it on}. In most cases, it is necessary to equip \eqref{SF} with a modified initial condition; see~\citet{Wechselberger2020} for details. For a planar system, $(m=n=1)$ with a $1$-dimensional critical submanifold, the modified initial condition is found by finding the fast fiber, $\mathcal{F}_0$, that contains the initial condition $(x_0,y_0)$. The initial condition assigned to \eqref{SF} is simply the base point $p_b\in S_0$, of the fast fiber that contains $(x_0,y_0)$.

For low dimensional systems (consider, for example $(z,u)\in \mathbb{R}^2$), it is tractable to transform (this is guaranteed at least locally) the perturbed differential equation into the {\it standard form}:
\begin{subequations}
\begin{align}
\dot{z} &= \varepsilon h(z,u,\varepsilon),\\
\dot{u} &= \;w(z,u,\varepsilon),
\end{align}
\end{subequations}
in which case the critical manifold $S_0$ is 
\begin{equation*}
S_0 :=\{(z,u)\in \mathbb{R}^2 : w(z,y,0)=0\},
\end{equation*}
and $z$ is called the {\it slow variable} since its velocity vanishes completely when $\varepsilon =0$. As long as $\partial_uw(z,u,0)\neq 0$ for all $(z,u)\in S_0$, the implicit function theorem ensures the existence of a smooth function, $u=\mu(z)$, such that $w(z,\mu(z))=0$. The dynamics on $S_{\varepsilon}$ is approximately
\begin{subequations}
\begin{align}
z' &= h(z,\mu(z),0),\label{autZ}\\
u' &= \mu'(z)h(z,\mu(z),0), \qquad \mu'(z) \equiv \cfrac{{\rm d}\mu}{{\rm d}z}.
\end{align}
\end{subequations}
Furthermore, the fiber containing the initial condition $(z(0),u(0))$ intersects the critical manifold at a point $(z(0),\cdot)$ since $z$ is constant along the fast fiber. Thus, we can assign \eqref{autZ} with $z(0)$. 

For two-dimensional planar systems, $\dot{x} = A(\varepsilon)x$ it is not necessary to diagonalize $A(0)$ to identify the slow variable. In the standard form, one will often have
\begin{equation}
\begin{pmatrix}\dot{x}\\\dot{y}\end{pmatrix} = \begin{pmatrix}\varepsilon a_{11}& \varepsilon a_{12}\\ a_{21} & -a_{22}\end{pmatrix}\begin{pmatrix}x\\y\end{pmatrix}
\end{equation}
and thus $x$ is a slow variable since its velocity vanishes with $\varepsilon =0$ and hence the first integral of the layer problem is simply $x$:
\begin{equation*}
x = (1\;\;0)\begin{pmatrix}x\\y\end{pmatrix}.
\end{equation*}
If we were to perform a coordinate transformation, $\mathcal{T}(x,y)$ to diagonalize $A(0)$, we have
\begin{equation*}
\mathcal{T}: \begin{pmatrix}x\\y\end{pmatrix}\mapsto \begin{pmatrix}x\\ y-\cfrac{a_{21}}{a_{22}}x\end{pmatrix}
\end{equation*}
and therefore $x$ is invariant with respect to diagonalization. The condition \eqref{21} reflects this invariance. 

\subsection*{A multiscale analysis via matched asymptotics}
In this subsection, we revisit the claim that the reduced CME,
\begin{equation}\label{RED5}
\cfrac{\partial}{\partial t}Pr(n,t) = \cfrac{\varepsilon k_3k_1}{k_1+k_2}(E_n^{+1}-1)nPr(n,t)
\end{equation}
adapted from the singular perturbation reduction of
\begin{subequations}
\begin{align*}
\dot{x} &= -k_1 x + k_2y,\\
\dot{y} &= \;\;\;k_1x - k_2y - \varepsilon k_3y,
\end{align*}
\end{subequations}
will not accurately approximate the timecourse variance of $n$, the molecular copy numbers of species ``$\bar{X}$" from \eqref{reac1}. Since the network is linear, the first two moments of the LNA and the full CME will agree, and therefore we can simply analyze the LNA to support the claim. Let us start with the layer problem of the mass action equations,
\begin{subequations}
\begin{align}
\dot{x} &= -k_1 x + k_2y,\\
\dot{y} &= \;\;\;k_1x - k_2y,
\end{align}
\end{subequations}
from which we see that the critical manifold (center subspace) is given by
\begin{equation*}
S_0 :=\{(x,y)\in \mathbb{R}^2: k_1x=k_2y\},\qquad A_0 = \begin{pmatrix}-k_1&\;\;k_2\\\;\;k_1&-k_2\end{pmatrix}.
\end{equation*}
The Jacobian, $A_0$, consists of a single trivial eigenvalue and one negative and nontrivial eigenvalue, $\lambda_-^{(0)}=-(k_1+k_2),$ and therefore $S_0$ is normally hyperbolic and attracting. The projection matrix, $P_0$, is 
\begin{equation*}
P_0 = I-\cfrac{1}{\lambda_-^{(0)}}A_0 = \cfrac{1}{k_1+k_2}\begin{pmatrix}k_2 & k_2\\k_1& k_1\end{pmatrix},
\end{equation*}
and therefore the reduced problem on the slow timescale is
\begin{equation}\label{redD}
\begin{pmatrix}x'\\y'\end{pmatrix} = -\cfrac{1}{k_1+k_2}\begin{pmatrix}k_2 & k_2\\k_1& k_1\end{pmatrix}\begin{pmatrix}0\\k_3y\end{pmatrix}\bigg|_{y=k_1x/k_2} = -\cfrac{k_1k_3x}{k_1+k_2}\begin{pmatrix}1\\ \cfrac{k_1}{k_2}\end{pmatrix}.
\end{equation}

It is important to understand that the reduced equation \eqref{redD} approximates the dynamics {\it on} the slow manifold (slow eigenspace); it does not account for any transient dynamics in the approach to the slow manifold. To approximate the fast dynamics, we solve the layer problem
\begin{equation}
\begin{pmatrix}
\dot{x}\\\dot{y} 
\end{pmatrix} = e^{\displaystyle tA_0}\begin{pmatrix}
x(0)\\y(0) 
\end{pmatrix}.
\end{equation}
The limiting behavior of the layer problem coincides with the initial condition to be supplied to \eqref{redD}:
\begin{equation*}
\lim_{t\to \infty} e^{\displaystyle tA_0}\begin{pmatrix} x(0)\\y(0) 
\end{pmatrix} = P_0\begin{pmatrix} x(0)\\y(0) 
\end{pmatrix} =x(0)\begin{pmatrix}1\\ \cfrac{k_1}{k_2}\end{pmatrix}.
\end{equation*}
The composite approximation is constructed by weaving together the layer and reduced problems
\begin{equation}
\begin{pmatrix}x\\y\end{pmatrix}(t,\tau) = e^{\displaystyle \lambda_-^{(0)}t}\begin{pmatrix}x(0)\\y(0)\end{pmatrix} + x(0)e^{\displaystyle \lambda_+^{(1)}\tau}\begin{pmatrix}1\\ \cfrac{k_1}{k_2}\end{pmatrix}- x(0)\begin{pmatrix}1\\ \cfrac{k_1}{k_2}\end{pmatrix}
\end{equation}
and approximates the exact solution across both the fast and slow timescales. Note that if we prescribe the initial condition to lie on $S_0$, then the solution reduces to
\begin{equation}
\begin{pmatrix}x\\y\end{pmatrix}(\tau) = x(0)e^{\displaystyle \lambda_+^{(1)}\tau}\begin{pmatrix}1\\ \cfrac{k_1}{k_2}\end{pmatrix},
\end{equation}
and depends only on the slow time. 

Now let us move on to the second moment, the variance, which obeys the matrix differential equation,
\begin{equation}\label{Clayer}
\dot{C} = \mathcal{L}_{A_0}(C)+\varepsilon \mathcal{L}_{A_1}(C) + Q_0(x,y) + \varepsilon Q_1(x,y)
\end{equation}
The layer problem is again obtained by setting $\varepsilon =0$ in \eqref{Clayer}. This results in
\begin{multline}
\cfrac{{\rm d}}{{\rm d}t}\begin{pmatrix}{\rm Var}(X) & {\rm Cov}(X,Y)\\{\rm Cov}(X,Y)& {\rm Var}(Y)\end{pmatrix} =\\  \begin{pmatrix}-k_1&\;\;k_2\\\;\;k_1&-k_2\end{pmatrix}\begin{pmatrix}{\rm Var}(X) & {\rm Cov}(X,Y)\\{\rm Cov}(X,Y)& {\rm Var}(X)\end{pmatrix}+\begin{pmatrix}{\rm Var}(X) & {\rm Cov}(X,Y)\\{\rm Cov}(X,Y)& {\rm Var}(X)\end{pmatrix}\begin{pmatrix}-k_1&\;\;k_1\\\;\;k_2&-k_2\end{pmatrix} +\\ (k_1x+k_2y)\begin{pmatrix}\;\;\;1&-1\\-1&\;\;\;1\end{pmatrix}.
\end{multline}
To simplify things, we will assume that the mass action equations are equipped with an initial condition that lies on $S_0$, so we will take $y(0)=k_1x(0)/k_2$. Moreover, we will assume that the initial state of the system is known with exact certainty; these assumptions leave us with the following:
\begin{equation}
Q_0(x,y) = 2k_1x(0)\begin{pmatrix}\;\;\;1&-1\\-1&\;\;\;1\end{pmatrix}, \qquad \quad \begin{pmatrix}{\rm Var}(X) & {\rm Cov}(X,Y)\\{\rm Cov}(X,Y)& {\rm Var}(Y)\end{pmatrix} (0)= \begin{pmatrix}0&0\\0&0\end{pmatrix}.
\end{equation}
The critical manifold for this problem, $\bar{S}_0$, is defined by the set
\begin{equation*}
\bar{S}_0 :=\{M \in S_2(\mathbb{R}): \mathcal{L}_{A_0}(M)=-Q_0(x,y)\},
\end{equation*}
where again $S_2(\mathbb{R})$ denotes the set of real, symmetric $2\times 2$ matrices, and is expressed parametrically as
\begin{equation*}
\bar{S}_0 := \cfrac{k_1x+k_2y}{2(k_1+k_2)}\begin{pmatrix}\;\;\;1&-1\\-1&\;\;\;1\end{pmatrix}+\alpha\begin{pmatrix}1&\cfrac{k_1}{k_2}\\\cfrac{k_1}{k_2}&\cfrac{k_1^2}{k_2^2}\end{pmatrix},\quad \alpha\in \mathbb{R}.
\end{equation*}
\begin{remark}
Observe that $\bar{S}_0$ is not simply $\ker \mathcal{L}_{A_0}$. This is a direct result of $Q_0(x,y)$ being nontrivial; again, see~\citet{JEWS2026}. If diffusion is limited to the slow timescale only, then $Q_0(x,y)$ is trivial and $\bar{S}_0$ will be the subspace $\bar{S}_0=\ker \mathcal{L}_{A_0}(\cdot)$.
\end{remark}

Because we have fixed $(x,y)(0)$ to lie on $S_0$, the solution to the layer problem is straightforward to compute
\begin{equation}
C(t) = \left(1-e^{\displaystyle 2\lambda_-^{(0)}t}\right)\cfrac{k_1x(0)}{k_1+k_2}\begin{pmatrix}\;\;\;1&-1\\-1&\;\;\;1\end{pmatrix},\quad \text{with} \quad \lim_{t\to \infty}C(t) = \cfrac{k_1x(0)}{k_1+k_2}\begin{pmatrix}\;\;\;1&-1\\-1&\;\;\;1\end{pmatrix}.
\end{equation}

Now let us move on to the reduced problem and let $\varepsilon >0$. On the slow timescale, the reduced equation for $x$ is
\begin{equation}
x(\tau) = x(0)e^{\displaystyle \lambda_+^{(1)}\tau}.
\end{equation}
Since the initial condition $x(0)$ decays to zero on the slow timescale, we have to account for this by updating the solution to the layer problem since it depends explicitly on $x$:
\begin{equation}\label{106}
C(t,\tau) = \left(1-e^{\displaystyle 2\lambda_-^{(0)}t}\right)\cfrac{k_1x(0)e^{\displaystyle \lambda_+^{(1)}\tau}}{k_1+k_2}\begin{pmatrix}\;\;\;1&-1\\-1&\;\;\;1\end{pmatrix}.
\end{equation}
We choose to parameterize the entire critical manifold, $S_0 \times \bar{S}_0$ by $x$ and ${\rm Var}(X)$. Thus, we will let $y=k_1x/k_2$ and set $\alpha = {\rm Var}(X)$. To compute a reduction for $C$ that depends parametrically on $x$ and ${\rm Var}(X)$, we substitute the ansatz
\begin{equation*}
C(x,{\rm var}(X))= \cfrac{k_1x}{k_1+k_2}\begin{pmatrix}\;\;\;1&-1\\-1&\;\;\;1\end{pmatrix}+{\rm Var}(X)\begin{pmatrix}1&\cfrac{k_1}{k_2}\\\cfrac{k_1}{k_2}&\cfrac{k_1^2}{k_2^2}\end{pmatrix}
\end{equation*}
into the right hand side of \eqref{Clayer} and then multiply by $P_0$ from the left and $P_0^T$ from the right. This leaves us with
\begin{equation}\label{expC}
C'(x,{\rm Var}(X)) = 2\lambda_+^{(1)}{\rm Var}(X)\begin{pmatrix}1&\cfrac{k_1}{k_2}\\\cfrac{k_1}{k_2}&\cfrac{k_1^2}{k_2^2}\end{pmatrix} + \cfrac{k_1k_3x}{k_1+k_2}\begin{pmatrix}\cfrac{k_2}{k_1+k_2} & \cfrac{k_1}{k_1+k_2} \\\cfrac{k_1}{k_1+k_2} & \cfrac{k_1^2}{k_2(k_1+k_2)}\end{pmatrix}.
\end{equation}

Let us now put \eqref{expC} into the form \eqref{FC}. From \eqref{58} we have 
\begin{equation*}
   \begin{pmatrix}\cfrac{k_2}{k_1+k_2} & \cfrac{k_1}{k_1+k_2} \\\cfrac{k_1}{k_1+k_2} & \cfrac{k_1^2}{k_2(k_1+k_2)}\end{pmatrix} = \begin{pmatrix}\cfrac{{\rm d}x}{{\rm d}z} & \cfrac{{\rm d}y}{{\rm d}z} \\\cfrac{{\rm d}y}{{\rm d}z} & \cfrac{{\rm d}y}{{\rm d}z}\cfrac{{\rm d}y}{{\rm d}x}\end{pmatrix}. 
\end{equation*}
Next, by elementary calculus we have
\begin{equation*}
-\cfrac{{\rm d}x}{{\rm d}t} = \cfrac{k_1k_3x}{k_1+k_2}= -\cfrac{{\rm d}z}{{\rm d}t}\cfrac{{\rm d}x}{{\rm d}z},
\end{equation*}
and therefore
\begin{equation}
\cfrac{k_1k_3x}{k_1 +k_2}\begin{pmatrix}\cfrac{{\rm d}x}{{\rm d}z} & \cfrac{{\rm d}y}{{\rm d}z} \\\cfrac{{\rm d}y}{{\rm d}z} & \cfrac{{\rm d}y}{{\rm d}z}\cfrac{{\rm d}y}{{\rm d}x}\end{pmatrix}  = |\lambda_+^{(0)}|z\begin{pmatrix}\cfrac{{\rm d}x}{{\rm d}z}\cfrac{{\rm d}x}{{\rm d}z} & \cfrac{{\rm d}y}{{\rm d}z}\cfrac{{\rm d}x}{{\rm d}z} \\\cfrac{{\rm d}y}{{\rm d}z}\cfrac{{\rm d}x}{{\rm d}z} & \cfrac{{\rm d}y}{{\rm d}z}\cfrac{{\rm d}y}{{\rm d}z}\end{pmatrix}.
\end{equation}
With $\varepsilon = 0$, we have the following LNA for $X$ and $Y$ (the fluctuations around $x$ and $y$, respectively),
\begin{equation*}
\begin{pmatrix}{\rm d}X\\{\rm d}Y\end{pmatrix} = \begin{pmatrix}-k_1& \;\;k_2\\\;\;\;k_1 & -k_2\end{pmatrix}\begin{pmatrix}X\\Y\end{pmatrix} + \begin{pmatrix}-\sqrt{k_1x}& \sqrt{k_2y}\\\sqrt{k_1x}&-\sqrt{k_2y}\end{pmatrix}\begin{pmatrix}{\rm d}W_1(t)\\{\rm d}W_2(t)\end{pmatrix}.
\end{equation*}
The sum ${\rm d}X+{\rm d}Y=0$, and therefore the slow variable is $Z=X+Y$, where
\begin{equation}
{\rm Var}(Z) = {\rm Var}(X)+{\rm Var}(Y) + 2{\rm Cov}(X,Y).
\end{equation}
However, on $\bar{S}_0$ we have
\begin{equation*}
{\rm Var}(Y) = \cfrac{k_1^2}{k_2^2}{\rm Var}(X) = \left(\cfrac{{\rm d}y}{{\rm d}x}\right)^2{\rm Var}(X), \quad {\rm Cov}(X,Y) = \cfrac{k_1}{k_2}{\rm Var}(X) = \cfrac{{\rm d}y}{{\rm d}x}{\rm Var}(X),
\end{equation*}
and therefore we can reparameterize ${\rm Var}(X)$ in terms of ${\rm Var}(Z)$ as follows:
\begin{equation}
{\rm Var}(Z) = \left(1+\cfrac{{\rm d}y}{{\rm d}x}\right)^2{\rm Var}(X) =\left(\cfrac{k_2+k_1}{k_2}\right)^2{\rm Var}(X) = \left(\cfrac{{\rm d}z}{{\rm d}x}\right)^2{\rm Var}(X).
\end{equation}
Finally, we have 
\begin{multline}
2\lambda_+^{(1)}{\rm Var}(X)\begin{pmatrix}1&\cfrac{k_1}{k_2}\\\cfrac{k_1}{k_2}&\cfrac{k_1^2}{k_2^2}\end{pmatrix} = 2\lambda_+^{(1)}\left(\cfrac{{\rm d}x}{{\rm d}z}\right)^2{\rm Var}(Z)\begin{pmatrix}1 & \cfrac{{\rm d}y}{{\rm d}x}\\ \cfrac{{\rm d}y}{{\rm d}x}& \left(\cfrac{{\rm d}y}{{\rm d}x}\right)^2\end{pmatrix} \\= 2\lambda_+^{(1)}{\rm Var}(Z)\begin{pmatrix}\cfrac{{\rm d}x}{{\rm d}z}\cfrac{{\rm d}x}{{\rm d}z} & \cfrac{{\rm d}y}{{\rm d}z}\cfrac{{\rm d}x}{{\rm d}z} \\\cfrac{{\rm d}y}{{\rm d}z}\cfrac{{\rm d}x}{{\rm d}z} & \cfrac{{\rm d}y}{{\rm d}z}\cfrac{{\rm d}y}{{\rm d}z}\end{pmatrix},
\end{multline}
and therefore the Fenichel reduction of $C$ is formally
\begin{equation}\label{FRF}
C'({\rm Var}(Z),z) = \lambda_+^{(1)}\left[2{\rm Var}(Z) -z\right]\begin{pmatrix}\cfrac{{\rm d}x}{{\rm d}z}\cfrac{{\rm d}x}{{\rm d}z} & \cfrac{{\rm d}y}{{\rm d}z}\cfrac{{\rm d}x}{{\rm d}z} \\\cfrac{{\rm d}y}{{\rm d}z}\cfrac{{\rm d}x}{{\rm d}z} & \cfrac{{\rm d}y}{{\rm d}z}\cfrac{{\rm d}y}{{\rm d}z}\end{pmatrix},\quad \text{with}\;\;z(\tau) = z(0)e^{\displaystyle \lambda_+^{(1)}\tau}
\end{equation}
the solution t which is
\begin{equation}\label{113}
C(\tau) = \left[{\rm Var}(Z)(0)e^{\displaystyle 2\lambda_+^{(1)}\tau} + z(0)\left(e^{\displaystyle \lambda_+^{(1)}\tau}-e^{\displaystyle 2\lambda_+^{(1)}\tau}\right)\right]\begin{pmatrix}\cfrac{{\rm d}x}{{\rm d}z}\cfrac{{\rm d}x}{{\rm d}z} & \cfrac{{\rm d}y}{{\rm d}z}\cfrac{{\rm d}x}{{\rm d}z} \\\cfrac{{\rm d}y}{{\rm d}z}\cfrac{{\rm d}x}{{\rm d}z} & \cfrac{{\rm d}y}{{\rm d}z}\cfrac{{\rm d}y}{{\rm d}z}\end{pmatrix}.
\end{equation}
To construct the composite approximation, $C_{\varepsilon}(t,\tau)$, we simply sum \eqref{113} and \eqref{106}, which reduces to 
\begin{equation}\label{CompCS}
C_{\varepsilon}(t,\tau)=\left(1-e^{\displaystyle 2\lambda_-^{(0)}t}\right)\cfrac{k_1x(0)e^{\displaystyle \lambda_+^{(1)}\tau}}{k_1+k_2}\begin{pmatrix}\;\;\;1&-1\\-1&\;\;\;1\end{pmatrix} + \left[ z(0)\left(e^{\displaystyle \lambda_+^{(1)}\tau}-e^{\displaystyle 2\lambda_+^{(1)}\tau}\right)\right]\begin{pmatrix}\cfrac{{\rm d}x}{{\rm d}z}\cfrac{{\rm d}x}{{\rm d}z} & \cfrac{{\rm d}y}{{\rm d}z}\cfrac{{\rm d}x}{{\rm d}z} \\\cfrac{{\rm d}y}{{\rm d}z}\cfrac{{\rm d}x}{{\rm d}z} & \cfrac{{\rm d}y}{{\rm d}z}\cfrac{{\rm d}y}{{\rm d}z}\end{pmatrix}
\end{equation}
under the assumption that the initial covariance is zero.
Observe that the derivative of the solution \eqref{CompCS} with respect to time is not $\mathcal{O}(\varepsilon)$, despite the fact that \eqref{CompCS} was constructed on the assumption that the initial condition of the mean field (mass action system) lies exactly on $S_0$:
\begin{equation}\label{Cder}
\cfrac{{\rm d}}{{\rm d}t}C_{\varepsilon}(t,\varepsilon t)= -2\lambda_-^{(0)}\cdot\cfrac{k_1x(0)e^{\displaystyle \lambda_+^{(1)}\tau}}{k_1+k_2}\begin{pmatrix}\;\;\;1&-1\\-1&\;\;\;1\end{pmatrix} + \mathcal{O}(\varepsilon).
\end{equation}
The implication here is that although the system was prepared to lie exactly on $S_0$ at time $t=0$, the initial velocity of $C$ is positive and $\mathcal{O}(1)$, which means that we expect to see an initial rapid growth in variance over the fast timescale. This is confirmed numerically; see {\sc figure} \ref{FIG9}.
\begin{figure}[htb!]
    \centering
    \includegraphics[scale=0.50]{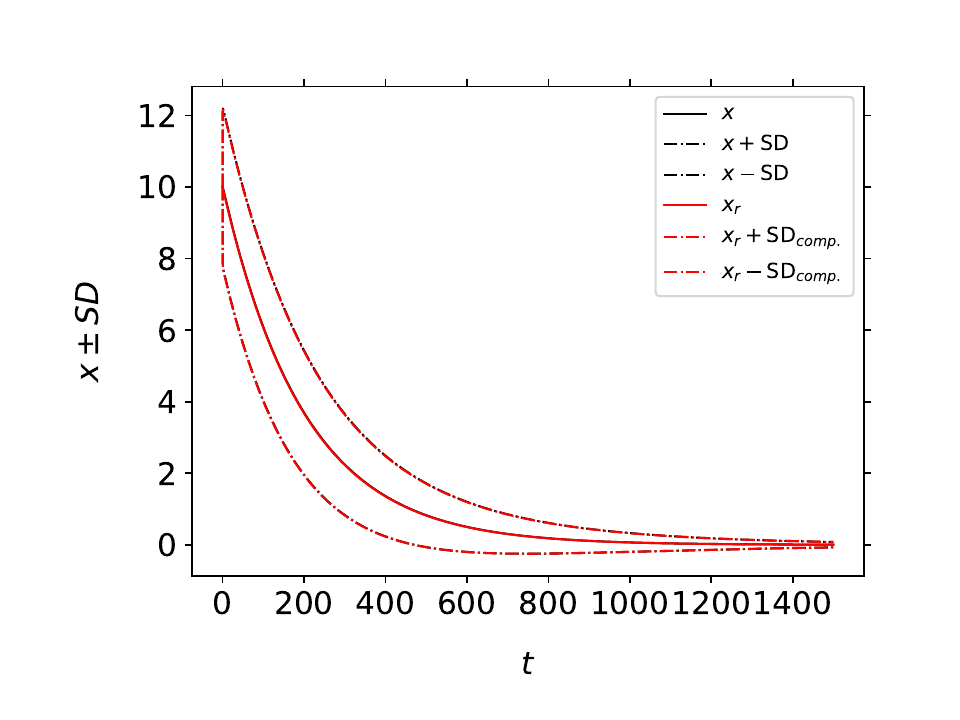}
    \includegraphics[scale=0.50]{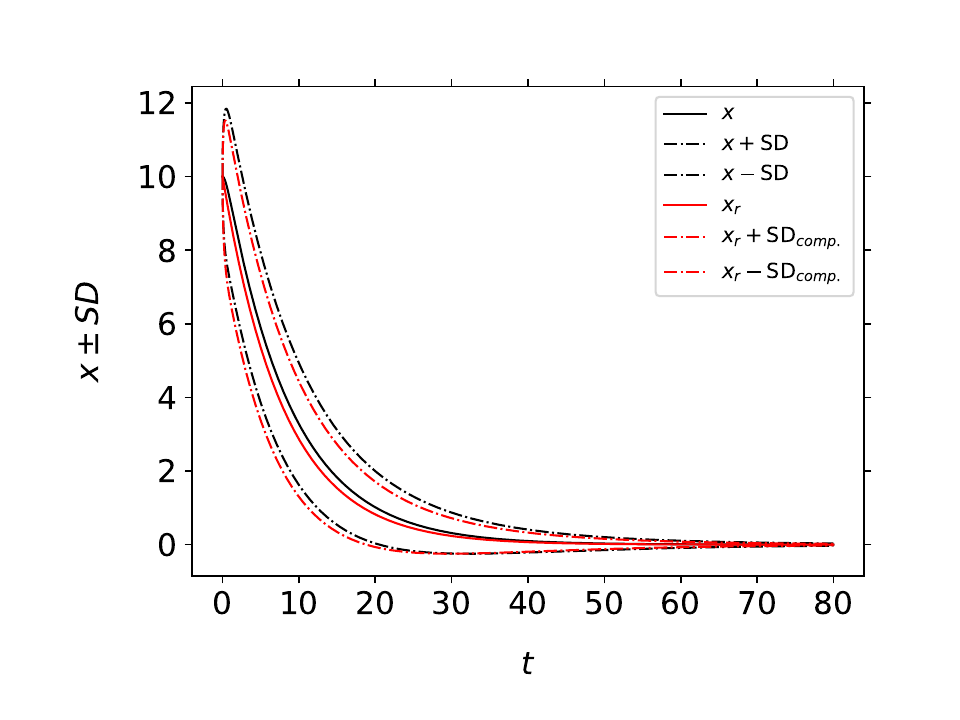}
    \caption{\textbf{The composite solution \eqref{CompCS} provides a highly accurate approximation to the variance as $k_3\to 0$.} In this figure, the solid black line is the mean of $x$ obtained by numerically solving the mass action equations; the dashed/dotted black lines are $\pm$ one standard deviation obtained from the numerical solution to \eqref{Clayer}. The solid red line is the numerical solution to the reduced equation for $x$ given by \eqref{redD} (this solution is by denoted $x_r$ in the legend). The dashed/dotted red curves are $\pm$ one standard deviation obtained from the composite solution \eqref{CompCS}. Observe that although the $C(0)=0^{2\times2}$, there is an initial transient period in which the variance rapidly increases due to the $\mathcal{O}(1)$ term in \eqref{Cder}. Compare this with {\sc figure} \eqref{FIG7}, where the initial increase in variance is slow and which creates a blunted plot near $t=0$. {\sc left panel}: In this simulation, $x(0)=10,y(0)=k_1x(0)/k_2$ and (in arbitrary units) $k_1=k_2=1.0$ and $k_3=0.25$. {\sc right panel:} In this simulation, $x(0)=10,y(0)=k_1x(0)/k_2$ and (in arbitrary units) $k_1=k_2=1.0$ and $k_3=0.01$.}
    \label{FIG9}
\end{figure}

The fact that the velocity of $C(t)$ is not $\mathcal{O}(\varepsilon)$ for $t\geq 0$ is a consequence of Proposition \ref{prop6}: Since $x$ not a slow variable, and since $S_0$ does not coincide with the $x$-axis, diffusion will occur on the fast timescale, even if the mass action equations are constrained to start on $S_0$.

The Fenichel (quasi-steady-state) reduction of the LNA for $x$ from \eqref{FRF} is (compare with \eqref{LN1})
\begin{equation*}
{\rm d}X = -\cfrac{k_3k_1}{k_1+k_2}\cdot X\;{\rm d}\tau- \sqrt{\cfrac{k_3k_1}{k_1+k_2}\cdot\textcolor{black}{\left(\cfrac{k_2}{k_1+k_2}\right)}\cdot x}\;{\rm d}W(\tau)
\end{equation*}
{\it sub-Poissonian} and different from the LNA of the heuristically reduced CME \eqref{RED5}, and it contains no indication that diffusion occurs over fast and slow timescales. Thus, it is impossible for a CME of the form
\begin{equation}
Pr(n,\tau) = \begin{pmatrix}N\\n\end{pmatrix}e^{-\displaystyle n\lambda_+^{(1)}\tau}\left(1-e^{-\displaystyle \lambda_+^{(1)}\tau}\right)^{(N-n)}
\end{equation}
to accurately approximate the variance of $n$ of the full CME since it lacks fast timescale dependence. 

\subsection*{Reductions of the LNA to the Michaelis-Menten reaction mechanism at high substrate concentration}
In this subsection we provide the details concerning the reduction of
\begin{multline}\label{ssL1}
\begin{pmatrix}
{\rm d}X_s\\{\rm d}X_c
\end{pmatrix} = \begin{pmatrix}\;\;\;k_1c & \;\;\;k_1(s+K_S)\\-k_1c & -k_1(s+K_M)
\end{pmatrix}\begin{pmatrix}
X_s\\X_c \end{pmatrix}{\rm d}t+ \begin{pmatrix}-k_1e_0\\\;\;\;k_1e_0 
\end{pmatrix}X_s\;{\rm d}\tau \\+ \begin{pmatrix}-\sqrt{k_1(\varepsilon e_0-c)s}& \sqrt{k_{-1}c}&0\\\sqrt{k_1(\varepsilon e_0-c)s}&-\sqrt{k_{-1}c}&-\sqrt{k_2c}\end{pmatrix}{\rm d}{\bf W}(t).
\end{multline}
To begin, the first-order approximation to the slow manifold is well-known
\begin{equation*}
c= \cfrac{\varepsilon e_0s}{s+K_M} = \varepsilon e_0 + \mathcal{O}(\varepsilon \varepsilon_2).
\end{equation*}
Insertion of $c=\varepsilon e_0$ into \eqref{ssL1} yields
\begin{equation}\label{ssL2}
\begin{pmatrix}
{\rm d}X_s\\{\rm d}X_c
\end{pmatrix} = \begin{pmatrix} 0&\;\;\;k_1(s+K_S)\\0& -k_1(s+K_M)
\end{pmatrix}\begin{pmatrix}
X_s\\X_c \end{pmatrix}\;{\rm d}t+  \sqrt{\varepsilon}\begin{pmatrix}0& \sqrt{k_{-1}e_0}&0\\0&-\sqrt{k_{-1}e_0}&-\sqrt{k_2e_0}\end{pmatrix}{\rm d}{\bf W}(t).
\end{equation}
Next, expand the drift matrix asymptotically
\begin{equation*}
\begin{pmatrix} 0&\;\;\;k_1(s+K_S)\\0& -k_1(s+K_M)
\end{pmatrix} = k_1s\begin{pmatrix} 0&\;\;\;1\\0& -1
\end{pmatrix} + \mathcal{O}(\varepsilon_1,\varepsilon_2),
\end{equation*}
which gives us
\begin{equation}
\begin{pmatrix}
{\rm d}X_s\\{\rm d}X_c
\end{pmatrix}  = k_1s\begin{pmatrix} 0&\;\;\;1\\0& -1
\end{pmatrix}X_c\;{\rm d}t + \mathcal{O}(\varepsilon_1,\varepsilon_2)X_c\;{\rm d}t + \begin{pmatrix}0& \sqrt{k_{-1}e_0}&0\\0&-\sqrt{k_{-1}e_0}&-\sqrt{k_2e_0}\end{pmatrix}{\rm d}{\bf W}(\tau)
\end{equation}
The projection matrix $P_0$ is straightforward to compute
\begin{equation}
P_0 = \begin{pmatrix}1&1\\0&0\end{pmatrix}:\mathbb{R}^2 \to \ker \begin{pmatrix}0 & \;\;\;1\\0&-1\end{pmatrix}.
\end{equation}
Moreover, the critical manifold coincides with the $X_s$ coordinate axis, and therefore we set $X_c=0$. The projection of the diffusion term onto the kernel of the drift matrix is
\begin{equation}
\begin{pmatrix}
{\rm d}X_s\\{\rm d}X_c
\end{pmatrix} = \begin{pmatrix}-\sqrt{k_2e_0}\\0\end{pmatrix}\;{\rm d}W(\tau).
\end{equation}
On a final note, it is important to acknowledge that the drift is only zero up to leading order. In fact, if we include higher order terms in $\varepsilon_1,\varepsilon_2$ we obtain
\begin{equation*}
{\rm d}X_s = -\cfrac{k_2e_0}{s}\cdot X_s\;{\rm d}\tau -\sqrt{k_2e_0}\;{\rm d}W(\tau),
\end{equation*}
but the term $k_2/s\sim \mathcal{O}(\varepsilon_2)$, and therefore when $s$ is sufficiently large, the drift evolves on a much slower timescale than the diffusion. Thus, to leading order, the fluctuations  satisfy a pure diffusion process. 
\newpage
\bibliographystyle{elsarticle-num-names}
\bibliography{CME.bib}

@article {HekGSPT,
    AUTHOR = {Hek, Geertje},
     TITLE = {Geometric singular perturbation theory in biological practice},
   JOURNAL = {J. Math. Biol.},
  FJOURNAL = {Journal of Mathematical Biology},
    VOLUME = {60},
      YEAR = {2010},
    NUMBER = {3},
     PAGES = {347--386},
}

@article{GrimaSecondOrder,
  title = {Linear-noise approximation and the chemical master equation agree up to second-order moments for a class of chemical systems},
  author = {Grima, Ramon},
  journal = {Phys. Rev. E},
  volume = {92},
  issue = {4},
  pages = {042124},
  numpages = {10},
  year = {2015},
}

@misc{greenbaum2019,
      title={First-order Perturbation Theory for Eigenvalues and Eigenvectors}, 
      author={Anne Greenbaum and Ren-Cang Li and Michael L. Overton},
      year={2019},
      eprint={1903.00785},
      archivePrefix={arXiv},
      primaryClass={math.NA},
      url={https://arxiv.org/abs/1903.00785}, 
}

@article{ganguly2025asymptotic,
    AUTHOR = {Ganguly, Arnab and KhudaBukhsh, Wasiur R.},
     TITLE = {Asymptotic analysis of the total quasi-steady state
              approximation for the {M}ichaelis-{M}enten enzyme kinetic
              reactions},
   JOURNAL = {J. Math. Anal. Appl.},
  FJOURNAL = {Journal of Mathematical Analysis and Applications},
    VOLUME = {561},
      YEAR = {2026},
    NUMBER = {1},
     PAGES = {Paper No. 130551, 34}
}

@article{BerglundGentzJDE,
title = {Geometric singular perturbation theory for stochastic differential equations},
journal = {Journal of Differential Equations},
volume = {191},
number = {1},
pages = {1-54},
year = {2003},
author = {Nils Berglund and Barbara Gentz}
}

@book {Eldering,
    AUTHOR = {Eldering, Jaap},
     TITLE = {Normally hyperbolic invariant manifolds},
    SERIES = {Atlantis Studies in Dynamical Systems},
    VOLUME = {2},
      NOTE = {The noncompact case},
 PUBLISHER = {Atlantis Press, Paris},
      YEAR = {2013},
     PAGES = {xii+189},
      ISBN = {978-94-6239-002-7; 978-94-6239-003-4},
   MRCLASS = {37D10},
  MRNUMBER = {3098498}
}

@article {KnoblochCM,
    AUTHOR = {Knobloch, E. and Wiesenfeld, K. A.},
     TITLE = {Bifurcations in fluctuating systems: the center-manifold
              approach},
   JOURNAL = {J. Statist. Phys.},
  FJOURNAL = {Journal of Statistical Physics},
    VOLUME = {33},
      YEAR = {1983},
    NUMBER = {3},
     PAGES = {611--637},
}

@article {Popovic,
    AUTHOR = {Popovi\'c, Nikola and Marr, Carsten and Swain, Peter S.},
     TITLE = {A geometric analysis of fast-slow models for stochastic gene
              expression},
   JOURNAL = {J. Math. Biol.},
  FJOURNAL = {Journal of Mathematical Biology},
    VOLUME = {72},
      YEAR = {2016},
    NUMBER = {1-2},
     PAGES = {87--122}
}

@article{KUEHNstoch,
title = {A mathematical framework for critical transitions: Bifurcations, fast–slow systems and stochastic dynamics},
journal = {Physica D: Nonlinear Phenomena},
volume = {240},
number = {12},
pages = {1020-1035},
year = {2011},
author = {Christian Kuehn}
}

@article {kurtz1978,
    AUTHOR = {Kurtz, Thomas G.},
     TITLE = {Strong approximation theorems for density dependent {M}arkov
              chains},
   JOURNAL = {Stochastic Process. Appl.},
  FJOURNAL = {Stochastic Processes and their Applications},
    VOLUME = {6},
      YEAR = {1977/78},
    NUMBER = {3},
     PAGES = {223--240}
}

@article{Yin2005,
title = {Limit behavior of two-time-scale diffusions revisited},
journal = {Journal of Differential Equations},
volume = {212},
number = {1},
pages = {85-113},
year = {2005},
author = {R.Z. Khasminskii and G. Yin}
}

@book {Wiggins,
    AUTHOR = {Wiggins, Stephen},
     TITLE = {Normally hyperbolic invariant manifolds in dynamical systems},
    SERIES = {Applied Mathematical Sciences},
    VOLUME = {105},
      NOTE = {With the assistance of Gy\"orgy Haller and Igor Mezi\'c},
 PUBLISHER = {Springer-Verlag, New York},
      YEAR = {1994},
     PAGES = {x+193}
}

@article{Yin2004,
author = {Khasminskii, R. Z. and Yin, G.},
title = {On Averaging Principles: An Asymptotic Expansion Approach},
journal = {SIAM Journal on Mathematical Analysis},
volume = {35},
number = {6},
pages = {1534-1560},
year = {2004}
}

@article {Kang2019,
    AUTHOR = {Kang, Hye-Won and KhudaBukhsh, Wasiur R. and Koeppl, Heinz and
              Rempa{\l}a, Grzegorz A.},
     TITLE = {Quasi-steady-state approximations derived from the stochastic
              model of enzyme kinetics},
   JOURNAL = {Bull. Math. Biol.},
  FJOURNAL = {Bulletin of Mathematical Biology. A Journal Devoted to
              Research at the Interface of the Life and Mathematical
              Sciences},
    VOLUME = {81},
      YEAR = {2019},
    NUMBER = {5},
     PAGES = {1303--1336}
}

@article{Janssen1989,
title = {The elimination of fast variables in complex chemical reactions. {I}{I}{I}. Mesoscopic level},
fjournal = {Journal of Statistical Physics},
journal = {J. Stat. Phys.},
volume = {57},
pages = {187-198},
year = {1989},
author = {Janssen, J. A. M.}
}

@article{Burrage2008,
author = {MacNamara,Shev  and Bersani,Alberto M.  and Burrage,Kevin  and Sidje,Roger B. },
title = {Stochastic chemical kinetics and the total quasi-steady-state assumption: Application to the stochastic simulation algorithm and chemical master equation},
journal = {The Journal of Chemical Physics},
volume = {129},
number = {9},
pages = {095105},
year = {2008}
}

@article{Rao2003,
author = {Rao,Christopher V.  and Arkin,Adam P. },
title = {Stochastic chemical kinetics and the quasi-steady-state assumption: Application to the Gillespie algorithm},
fjournal = {The Journal of Chemical Physics},
journal = {J. Chem. Phys.},
volume = {118},
number = {11},
pages = {4999-5010},
year = {2003}
}

@incollection{VKX,
title = {Chapter {X}. {T}HE EXPANSION OF THE {M}ASTER {E}QUATION},
booktitle = {Stochastic Processes in Physics and Chemistry ($3^{\text{rd}}$ Edition)},
publisher = {Elsevier},
address = {Amsterdam},
pages = {244--272},
year = {2007},
series = {North-Holland Personal Library},
author = {N.G. Van Kampen}
}

@article{Thomas2012,
title={The slow-scale linear noise approximation: an accurate, reduced stochastic description of biochemical networks under timescale separation conditions},
author={Thomas, Philipp and Straube, Arthur V. and Grima, Ramon},
fjournal={BMC Systems Biology},
journal={BMC Sys. Biol.},
Volume={6},
Number={1},
pages={39},
year={2012}
}

@article{Fenichel1979,
    AUTHOR = {Fenichel, Neil},
     TITLE = {Geometric singular perturbation theory for ordinary
              differential equations},
   JOURNAL = {J. Differ. Equations},
  FJOURNAL = {Journal of Differential Equations},
    VOLUME = {31},
      YEAR = {1979},
     PAGES = {53--98}
}

@article{Fenichel1971,
    AUTHOR = {Fenichel, Neil},
     TITLE = {Persistence and smoothness of invariant manifolds for flows},
   JOURNAL = {Indiana Univ. Math. J.},
  FJOURNAL = {Indiana University Mathematics Journal},
    VOLUME = {21},
      YEAR = {1971/72},
     PAGES = {193--226}
}

@article{Borghans1996,
title={Extending the quasi-steady state approximation by changing variables},
  author={Borghans, Jos{\'e} A. M. and De Boer, Rob J. and Segel, Lee A.},
  fjournal={Bulletin of Mathematical Biology},
  journal = {Bull. Math. Biol.},
  volume={58},
  pages={43--63},
  year={1996},
  publisher={Springer}
}

@article{Segel1989,
   AUTHOR = {Segel, L. A. and Slemrod, M.},
    TITLE = {The Quasi-Steady-State Assumption: {A} case study in perturbation},
  JOURNAL = {SIAM Rev.},
 FJOURNAL = {SIAM Review. A publication of the {S}ociety for {I}ndustrial and
             {A}pplied {M}athematics},
   VOLUME = {31},
     YEAR = {1989},
    PAGES = {446--477},
    CODEN = {SIREAD},
  MRCLASS = {92B05 (34E15 92C45)},
 MRNUMBER = {92f:92003},
   MRREVR = {J. C. Misra}
}

@article{Heineken1967,
   AUTHOR = {Heineken, F. G. and Tsuchiya, H. M. and Aris, R.},
    TITLE = {On the mathematical status of the pseudo-steady hypothesis of
             biochemical kinetics},
  JOURNAL = {Math. Biosci.},
 FJOURNAL = {Mathematical Biosciences},
   VOLUME = {1},
     YEAR = {1967},
    PAGES = {95--113},
    CODEN = {MABIAR}
}

@article{Segel1988,
    AUTHOR = {Segel, Lee A.},
     TITLE = {On the validity of the steady state assumption of enzyme kinetics},
   JOURNAL = {Bull. Math. Biol.},
  FJOURNAL = {Bulletin of Mathematical Biology},
    VOLUME = {50},
      YEAR = {1988},
     PAGES = {579--593}
}

@article{Schnell1997,
  title={Closed form solution for time-dependent enzyme kinetics},
  author={Schnell, Santiago and Mendoza, Claudio},
  fjournal={Journal of Theoretical Biology},
  journal={J. Theor. Biol.},
  volume={187},
  pages={207--212},
  year={1997},
  publisher={Elsevier}
}

@article{Goeke2012,
    AUTHOR = {Goeke, Alexandra and Schilli, Christian and Walcher, Sebastian
              and Zerz, Eva},
     TITLE = {Computing quasi-steady state reductions},
   JOURNAL = {J. Math. Chem.},
  FJOURNAL = {Journal of Mathematical Chemistry},
    VOLUME = {50},
      YEAR = {2012},
     PAGES = {1495--1513}
}

@article{Goeke2015,
    AUTHOR = {Goeke, Alexandra and Walcher, Sebastian and Zerz, Eva},
     TITLE = {Determining ``small parameters'' for quasi-steady state},
   JOURNAL = {J. Differ. Equations.},
  FJOURNAL = {Journal of Differential Equations},
    VOLUME = {259},
      YEAR = {2015},
     PAGES = {1149--1180}
}

@InProceedings{Goeke2013,
 author={Goeke, Alexandra
 and Walcher, Sebastian},
 editor={Johann, Andreas
 and Kruse, Hans-Peter
 and Rupp, Florian
 and Schmitz, Stephan},
 title={Quasi-Steady State: Searching for and Utilizing Small Parameters},
 booktitle={Recent Trends in Dynamical Systems},
 year={2013},
 publisher={Springer Basel},
 address={Basel},
 pages={153--178}
}

@article{Noethen2011,
    AUTHOR = {Noethen, Lena and Walcher, Sebastian},
     TITLE = {Tikhonov's theorem and quasi-steady state},
   JOURNAL = {Discrete Contin. Dyn. Syst. Ser. B},
  FJOURNAL = {Discrete and Continuous Dynamical Systems Series B},
    VOLUME = {16},
      YEAR = {2011},
     PAGES = {945--961}
}

@book {kuehn2015,
    AUTHOR = {Kuehn, Christian},
     TITLE = {Multiple time scale dynamics},
    SERIES = {Applied Mathematical Sciences},
    VOLUME = {191},
 PUBLISHER = {Springer},
      YEAR = {2015},
     PAGES = {xiv+814},
   MRCLASS = {34-02 (34-01 34Exx 37-02 37C10 37Gxx)},
  MRNUMBER = {3309627},
MRREVIEWER = {Tewfik Sari}
}

@book{Wechselberger2020,
    AUTHOR = {Wechselberger, Martin},
     TITLE = {Geometric Singular Perturbation Theory Beyond the Standard Forms},
     SERIES = {Frontiers in Applied dynamical systems: Tutorials and Reviews},
     NUMBER = {6},
 PUBLISHER = {Springer},
      YEAR = {2020}
}

@article{Kim2020,
    author = {Kim, Jae Kyoung AND Tyson, John J.},
    fjournal = {PLoS Computational Biology},
    Journal = {PLoS Comp. Biol.},
    publisher = {Public Library of Science},
    title = {Misuse of the {M}ichaelis–-{M}enten rate law for protein interaction networks and its remedy},
    year = {2020},
    month = {10},
    volume = {16},
    pages = {1-21},
    number = {10}
}

@article{Agarwal2012,
author = {Agarwal,Animesh  and Adams,Rhys  and Castellani,Gastone C.  and Shouval,Harel Z. },
title = {On the precision of quasi steady state assumptions in stochastic dynamics},
journal = {The Journal of Chemical Physics},
volume = {137},
number = {4},
pages = {044105},
year = {2012}
}

@article{Thomas2011,
author = {Thomas, Philipp and Straube, Arthur V. and Grima, Ramon },
title = {Communication: Limitations of the stochastic quasi-steady-state approximation in open biochemical reaction networks},
fjournal = {The Journal of Chemical Physics},
journal = {J. Chem. Phys.},
volume = {135},
number = {18},
pages = {181103},
year = {2011}
}

@article{Holehouse,
author = {Holehouse,James  and Sukys,Augustinas  and Grima,Ramon },
title = {Stochastic time-dependent enzyme kinetics: Closed-form solution and transient bimodality},
fjournal = {The Journal of Chemical Physics},
journal={J. Chem. Phys.},
volume = {153},
number = {16},
pages = {164113},
year = {2020}
}

@article{Herath,
author = {Herath, Narmada  and Del Vecchio,Domitilla },
title = {Reduced linear noise approximation for biochemical reaction networks with time-scale separation: The stochastic t{Q}{S}{S}{A}$^+$},
journal = {J. Chem. Phys.},
fjournal = {The Journal of Chemical Physics},
volume = {148},
number = {9},
pages = {094108},
year = {2018}
}

@article{JSEssLNA,
	Author = {Eilertsen, Justin and Srivastava, Kashvi and Schnell, Santiago},
	Journal = {Journal of Mathematical Biology},
	Number = {1},
	Pages = {3},
	Title = {Stochastic enzyme kinetics and the quasi-steady-state reductions: Application of the slow scale linear noise approximation {\`a} la Fenichel},
	Volume = {85},
	Year = {2022},
 }

@book{RobertsBook,
    AUTHOR = {Roberts, A. J.},
     TITLE = {Model emergent dynamics in complex systems},
    SERIES = {Mathematical Modeling and Computation},
    VOLUME = {20},
 PUBLISHER = {Society for Industrial and Applied Mathematics (SIAM),
              Philadelphia, PA},
      YEAR = {2015},
     PAGES = {xii+748}
}

@article{ThomasPO,
  title = {Rigorous elimination of fast stochastic variables from the linear noise approximation using projection operators},
  author = {Thomas, Philipp and Grima, Ramon and Straube, Arthur V.},
  journal = {Phys. Rev. E},
  volume = {86},
  issue = {4},
  pages = {041110},
  numpages = {9},
  year = {2012}
}

@article{Parsons2017,
	year = 2017,
	volume = {50},
	number = {41},
	pages = {415601},
	author = {Todd L Parsons and Tim Rogers},
	title = {Dimension reduction for stochastic dynamical systems forced onto a manifold by large drift: a constructive approach with examples from theoretical biology},
	fjournal = {Journal of Physics A: Mathematical and Theoretical},
        Journal = {J. Phys. A Math. Theor.}
}

@article{KIM2015,
author = {Kim, Jae Kyoung and Josić, Krešimir and Bennett, Matthew R.},
title = {The relationship between stochastic and deterministic quasi-steady state approximations},
fjournal = {BMC Systems Biology},
journal = {BMC Syst. Biol.},
volume = {9},
pages = {87},
year = {2015}
}

@article{Mastny2007,
author = {Mastny, Ethan A.  and Haseltine, Eric L.  and Rawlings, James B. },
title = {Two classes of quasi-steady-state model reductions for stochastic kinetics},
fjournal = {The Journal of Chemical Physics},
journal = {J. Chem. Phys.},
volume = {127},
number = {9},
pages = {094106},
year = {2007}
}

@article{KangKim2017,
author = {Kim, Jae Kyoung and Rempala, Grzegorz A. and Kang, Hye-Won},
title = {Reduction for Stochastic Biochemical Reaction Networks with Multiscale Conservations},
fjournal = {Multiscale Modeling \& Simulation},
journal = {Multiscale. Model. Simul.},
volume = {15},
number = {4},
pages = {1376-1403},
year = {2017}
}

@article{KIM2014,
title = {The Validity of Quasi-Steady-State Approximations in Discrete Stochastic Simulations},
fjournal = {Biophysical Journal},
journal = {Biophys. J.},
volume = {107},
number = {3},
pages = {783 - 793},
year = {2014},
author = {Jae Kyoung Kim and Krešimir Josić and Matthew R. Bennett}
}

@article{GILLESPIECME,
title = {A rigorous derivation of the chemical master equation},
fjournal = {Physica A: Statistical Mechanics and its Applications},
journal = {Physica A},
author = {Gillespie, D. T.},
volume = {188},
pages = {404--425},
year = {1992}
}

@article{kang2013,
author = {Kang, Hye-Won and Kurtz, Thomas G.},
fjournal = {Annals of Applied Probability},
journal = {Ann. Appl. Probab.},
number = {2},
pages = {529--583},
publisher = "The Institute of Mathematical Statistics",
title = {Separation of time-scales and model reduction for stochastic reaction networks},
volume = {23},
year = {2013}
}

@article{Sanft,
author = {Sanft, K.R.  and Gillespie, D. T.  and Petzold, L. R. },
title = {The legitimacy of the stochastic {M}ichaelis--{M}enten approximation},
fjournal = {IET Systems Biology},
journal = {IET Syst. Biol.},
volume = {5},
pages = {58--69},
year = {2011},
}

@article{Tyson2008,
author ={Barik, Debashis and Paul, Mark R. and Baumann, William T. and Cao, Yang and Tyson, John J.},
title={Stochastic Simulation of Enzyme-Catalyzed Reactions with Disparate Timescales},
fjournal={Biophysical Journal},
journal={Biophys. J.},
year={2008},
volume={95},
number={8},
page={3563--3574}
}

@article{TURNER200,
title = {Stochastic approaches for modelling in vivo reactions},
journal = {Computational Biology and Chemistry},
volume = {28},
number = {3},
pages = {165-178},
year = {2004},
author = {T.E. Turner and S. Schnell and K. Burrage},
}

@article{PAHLAJANI201196,
title = {Stochastic reduction method for biological chemical kinetics using time-scale separation},
fjournal = {Journal of Theoretical Biology},
journal = {J. Theo. Biol.},
volume = {272},
number = {1},
pages = {96-112},
year = {2011},
author = {Chetan D. Pahlajani and Paul J. Atzberger and Mustafa Khammash}
}

@article{Katz,
    AUTHOR = {Katzenberger, G. S.},
     TITLE = {Solutions of a stochastic differential equation forced onto a
              manifold by a large drift},
   JOURNAL = {Ann. Probab.},
  FJOURNAL = {The Annals of Probability},
    VOLUME = {19},
      YEAR = {1991},
     PAGES = {1587--1628}
}

@article{Lax2020,
title = {Singular perturbations and scaling},
journal = {Discrete and Continuous Dynamical Systems - B},
volume = {25},
number = {1},
pages = {1-29},
year = {2020},
author = {Christian Lax and Sebastian Walcher}
}

@article{GrimaBreakdown,
  title = {Noise-Induced Breakdown of the {M}ichaelis-{M}enten Equation in Steady-State Conditions},
  author = {Grima, R.},
  journal = {Phys. Rev. Lett.},
  volume = {102},
  issue = {21},
  pages = {218103},
  numpages = {4},
  year = {2009},
}

@article {Roberts1996,
    AUTHOR = {Xu, Chao and Roberts, A. J.},
     TITLE = {On the low-dimensional modelling of {S}tratonovich stochastic
              differential equations},
   JOURNAL = {Phys. A},
  FJOURNAL = {Physica A},
    VOLUME = {225},
      YEAR = {1996},
    NUMBER = {1},
     PAGES = {62--80}
}

@article {Roberts2013,
    AUTHOR = {Wang, W. and Roberts, A. J.},
     TITLE = {Slow manifold and averaging for slow-fast stochastic
              differential system},
   JOURNAL = {J. Math. Anal. Appl.},
  FJOURNAL = {Journal of Mathematical Analysis and Applications},
    VOLUME = {398},
      YEAR = {2013},
    NUMBER = {2},
     PAGES = {822--839}
}

@article {Wilhelm2000,
    AUTHOR = {Schneider, Klaus R. and Wilhelm, Thomas},
     TITLE = {Model reduction by extended quasi-steady-state approximation},
   JOURNAL = {J. Math. Biol.},
  FJOURNAL = {Journal of Mathematical Biology},
    VOLUME = {40},
      YEAR = {2000},
    NUMBER = {5},
     PAGES = {443--450}
}

@article {OthmerI,
    AUTHOR = {Lee, Chang Hyeong and Othmer, Hans G.},
     TITLE = {A multi-time-scale analysis of chemical reaction networks.
              {I}. {D}eterministic systems},
   JOURNAL = {J. Math. Biol.},
  FJOURNAL = {Journal of Mathematical Biology},
    VOLUME = {60},
      YEAR = {2010},
    NUMBER = {3},
     PAGES = {387--450}
}

@article {FirstOrder,
    AUTHOR = {Eilertsen, Justin and Schnell, Santiago and Walcher,
              Sebastian},
     TITLE = {The {M}ichaelis-{M}enten reaction at low substrate
              concentrations: pseudo-first-order kinetics and conditions for
              timescale separation},
   JOURNAL = {Bull. Math. Biol.},
  FJOURNAL = {Bulletin of Mathematical Biology. A Journal Devoted to
              Research at the Interface of the Life and Mathematical
              Sciences},
    VOLUME = {86},
      YEAR = {2024},
    NUMBER = {6},
     PAGES = {Paper No. 68, 13}
}

@article{Lapuz,
author = {Lapuz, Timothy Earl Figueroa and Wechselberger, Martin},
title = {Coordinate-Independent Model Reductions of Chemical Reaction Networks Based on Geometric Singular Perturbation Theory},
journal = {SIAM Journal on Life Sciences},
volume = {1},
number = {2},
pages = {163-201},
year = {2026}
}

@Inbook{Jones1995,
author={Jones, Christopher K. R. T.},
title={Geometric singular perturbation theory},
bookTitle={Dynamical Systems: Lectures Given at the 2nd Session of the Centro Internazionale Matematico Estivo (C.I.M.E.) held in Montecatini Terme, Italy, June 13--22, 1994},
year={1995},
publisher={Springer Berlin Heidelberg},
address={Berlin, Heidelberg},
pages={44--118}
}

@article {JEWS2026,
    AUTHOR = {Eilertsen, Justin and Stroberg, Wylie},
     TITLE = {On the reduction of stochastic chemical reaction networks},
   JOURNAL = {J. Math. Biol.},
  FJOURNAL = {Journal of Mathematical Biology},
    VOLUME = {92},
      YEAR = {2026},
    NUMBER = {1},
     PAGES = {Paper No. 1, 35}
}

@inbook{Cracium,
author = {Pantea, Casian and Gupta, Ankur and Rawlings, James and Craciun, Gheorghe},
year = {2014},
month = {10},
pages = {419-442},
title = {The QSSA in Chemical Kinetics: As Taught and as Practiced},
journal = {Natural Computing Series}
}
\end{document}